\documentclass{revtex4-1}

\usepackage{amssymb,latexsym}
\usepackage{amsbsy} 
\usepackage{amsmath}
\usepackage{amsthm}

\usepackage{graphicx}
\usepackage{color}

\newtheorem{theorem}{Theorem}

\newcommand \trace[1] {\mathrm{Tr}({#1})}

\newcommand \op[1] {\hat{#1}}    

\newcommand \bra[1] {\langle {#1} |}
\newcommand \ket[1] {| {#1} \rangle}

\newcommand \MSobolev[1] {\mathcal{H}^1_{{#1}}}
\newcommand \WMSobolev[1] {\widetilde{\mathcal{H}}^1_{{#1}}}

\newcommand \DMset {\mathcal{D}}

\newcommand \ALinSpace {Y^*_{\text{lin}}}

\newcommand \jpvec {\mathbf{j}_{\mathrm{p}}}
\newcommand \jpvecix[1] {\mathbf{j}_{\mathrm{p};{#1}}}

\newcommand \jmvec {\mathbf{j}_{\mathrm{m}}}
\newcommand \jmvecix[1] {\mathbf{j}_{\mathrm{m};{#1}}}
\newcommand \jmcomp[1] {j_{\mathrm{m};{#1}}}

\newcommand \gspin {g_{\mathrm{s}}}

\newcommand \Fvr {F_{\text{VR}}}

\newcommand \Fvrdm {F_{\text{VR,DM}}}

\newcommand \Fldft {F}

\newcommand \Fldftdm {F_{\text{DM}}}

\newcommand \commentout[1] {}

\newcommand \magm[1] {{#1}}

\begin{document}

\title{Uniform magnetic fields in density-functional theory}

\author{Erik I. Tellgren}

\email{erik.tellgren@kjemi.uio.no}
\affiliation{
Hylleraas Centre for Quantum Molecular Sciences, Department of Chemistry, University of Oslo, P.O.~Box 1033 Blindern, N-0315 Oslo, Norway}

\author{Andre Laestadius}

\affiliation{
Hylleraas Centre for Quantum Molecular Sciences, Department of Chemistry, University of Oslo, P.O.~Box 1033 Blindern, N-0315 Oslo, Norway}

\author{Trygve Helgaker}

\affiliation{
Hylleraas Centre for Quantum Molecular Sciences, Department of Chemistry, University of Oslo, P.O.~Box 1033 Blindern, N-0315 Oslo, Norway}
\affiliation{Centre for Advanced Study at the Norwegian Academy of Science and Letters,
Drammensveien 78, N-0271 Oslo, Norway
}

\author{Simen Kvaal}

\affiliation{
Hylleraas Centre for Quantum Molecular Sciences, Department of Chemistry, University of Oslo, P.O.~Box 1033 Blindern, N-0315 Oslo, Norway}

\author{Andrew M. Teale}

\affiliation{
School of Chemistry, University of Nottingham, University Park, Nottingham NG7 2RD, UK}
\affiliation{Centre for Advanced Study at the Norwegian Academy of Science and Letters,
Drammensveien 78, N-0271 Oslo, Norway
}

\date{\today}

\begin{abstract}
We construct a density-functional formalism adapted to uniform external magnetic fields that is intermediate between conventional Density Functional Theory and Current-Density Functional Theory (CDFT). In the intermediate theory, which we term LDFT, the basic variables are the density, the canonical momentum, and the paramagnetic contribution to the magnetic moment. Both a constrained-search formulation and a convex formulation in terms of Legendre--Fenchel transformations are constructed. Many theoretical issues in CDFT find simplified analogues in LDFT. We prove results concerning $N$-representability, Hohenberg--Kohn-like mappings, existence of minimizers in the constrained-search expression, and a restricted analogue to gauge invariance. The issue of additivity of the energy over non-interacting subsystems, which is qualitatively different in LDFT and CDFT, is also discussed.
\end{abstract}

\maketitle

\section{Introduction}

Ground-state density functional theory (DFT) is a widely successful
electronic structure method due to its good trade off between accuracy and
computational cost. Although DFT is commonly used as a pragmatic
approach to compute the response of a molecule to an external magnetic
field, this application is formally outside the scope of the
theory. None of the independent theoretical foundations of DFT---that is, the
Hohenberg--Kohn (HK) theorem~\cite{HOHENBERG_PR136_864}, the
constrained-search formulation~\cite{LEVY_PNAS76_6062}, and Lieb's
convex analysis formulation~\cite{LIEB_IJQC24_243,ESCHRIG03}---directly allows
for a magnetic vector potential in the Hamiltonian. All three of these
frameworks establish, firstly, that the density-scalar potential
interaction can be isolated from other energy terms. Secondly, the
remaining energy terms can be obtained from a {\it universal} density
functional, which requires the density but not the external potential
as input. Some generalizations of DFT that allow for a magnetic
vector potential in the Hamiltonian have been proposed. Grayce and
Harris dropped the universality property and formulated magnetic-field
DFT, or BDFT, which can be viewed as a family of density-functional
theories parametrized by the external magnetic
field~\cite{GRAYCE_PRA50_3089,SALSBURY_JCP107_7350}. For a vanishing
magnetic field, the parametrized theory coincides with conventional
DFT. The resulting density functional is semi-universal---that is,
universal with respect to the scalar potential, but not with
respect to the magnetic field. The dominant alternative to BDFT that
preserves universality is the paramagnetic current-density functional
theory (CDFT) due to Vignale and Rasolt~\cite{VIGNALE_PRL59_2360,VIGNALE_PRB37_10685}. Here a universal functional
of the density and paramagnetic current density is established and
many results from conventional DFT carry over to this framework. Both
a constrained-search formulation and a generalization of Lieb's convex
formulation are possible for paramagnetic CDFT. The HK theorem of
conventional DFT has only a weaker analogue in CDFT,
which is sufficient for defining a ground-state energy CDFT functional
but not for generalization of all the stronger statements that can be
made based on the HK mapping.

In the present work, we explore a framework intermediate between
DFT and CDFT. By restricting attention to uniform magnetic
fields, we obtain a theory where universality does not require the
full paramagnetic current density as an additional variable besides the
density. Instead, the canonical momentum and the \magm{paramagnetic moment}
(equivalent to canonical angular momentum in the absence of spin) are
sufficient. The resulting theory is termed LDFT. The reduction from an
infinite-dimensional vector field to two three-dimensional quantities
represents a substantial simplification. In some respects it is also a
qualitative change, since the current density is a local quantity,
whereas linear momentum and \magm{magnetic moment} are global quantities. Both the
orbital and spin \magm{magnetic moments} contribute to the total \magm{magnetic
moment}. LDFT can therefore also be viewed as a minimal framework for
incorporating spin dependence in a universal density functional.

The outline of this paper is as follows. In Sec.~\ref{secCDFTREVIEW},
we review paramagnetic CDFT and adapt some of the technical
mathematical details to allow for uniform magnetic fields. We continue
in Sec.~\ref{secLDFT} by detailing the formulation of LDFT. Both
constrained-search and convex formulations are given. Some
HK-like results, which turn out to be stronger than the
CDFT analogues, are also explored along with issues such as
expectation-valuedness and $N$-representability. In Sec.~\ref{secKS}
we discuss Kohn--Sham theory, with focus on invariance with respect to
gauge degrees of freedom and additive separability of the
exchange--correlation functional. Next, in Sec.~\ref{secPSCRITIQUE}, we
briefly comment on the possibility of using gauge-invariant, physical
quantities instead of the paramagnetic current density or canonical
momenta. We also comment on a recent hybrid formulation due to Pan and
Sahni featuring both the canonical, gauge-dependent magnetic moment
and the physical current density~\cite{PAN_JCP143_174105}. Finally, in Sec.~\ref{secCONCLUSION}, we give
some concluding remarks.

\section{Review of paramagnetic current-density functional theory}
\label{secCDFTREVIEW}

Current-density functional theory is the natural generalization of DFT to the case when there is an external magnetic vector potential. For an $N$-electron system,
we consider Schr\"odinger--Pauli Hamiltonians of the type
\begin{equation}
  H(v,\mathbf{A}) = \frac{1}{2}  \sum_{k=1}^{N} \big(-i\nabla_k + \mathbf{A}(\mathbf{r}_k) \big)^2 + \frac{\gspin}{2}  \sum_{k=1}^N \mathbf{B}(\mathbf{r}_k) \cdot \op{\mathbf{S}}_k+ \sum_{k=1}^{N} v(\mathbf{r}_k) + \frac{1}{2} \sum_{k\neq l} \frac{1}{|\mathbf{r}_k - \mathbf{r}_l|},
\end{equation}
where atomic units are used, $\op{\mathbf{S}}_k = \frac{1}{2} \op{\boldsymbol{\sigma}}_k$ is a spin operator, $\mathbf{B}(\mathbf{r}_k) := \nabla_k\times \mathbf{A}(\mathbf{r}_k)$ is the magnetic field, and $\gspin$ is the electron spin $g$-factor. Although $\gspin = 2$ is the empirically relevant value for a non-relativistic theory, the CDFT formalism is mathematically valid for any value of $\gspin$. When we are only interested in analysing orbital effects and states that are ground states in the absence of the spin-Zeeman term, we are therefore free to set $\gspin = 0$. In what follows, we consider this parameter to be fixed but arbitrary.

We introduce the short-hand notation
\begin{equation}
  (f|g) = \int_{\mathbb{R}^3} f(\mathbf{r}) \, g(\mathbf{r}) d\mathbf{r}
\end{equation}
for the pairing integral of two scalar fields $f$ and $g$, and similarly for vector fields (where the pointwise product inside the integral is replaced by the scalar product).

Our point of departure is the constrained-search expression introduced by Vignale and Rasolt. For given external potentials $(v,\mathbf{A})$, we can write the
ground-state energy as
\begin{equation}
 \label{eqEvANAIVE}
  E(v,\mathbf{A}) = \inf_{\psi} \bra{\psi} H(v,\mathbf{A}) \ket{\psi}
  = \inf_{\rho,\jmvec} \Big( (\rho|v+\tfrac{1}{2} A^2) +
  (\jmvec|\mathbf{A}) + \inf_{\psi \mapsto \rho,\jmvec} \bra{\psi}
  -\frac{1}{2} \sum_k \nabla_k^2  + \frac{1}{2} \sum_{k\neq l} \frac{1}{|\mathbf{r}_k - \mathbf{r}_l|} \ket{\psi} \Big),
\end{equation}
where the standard notation $\psi \mapsto (\rho,\jmvec)$ means that the wave function gives rise to these densities---that is, $\rho_{\psi} = \rho$ and $\jmvecix{\psi} = \jmvec$. In detail, the densities corresponding to a wave function $\psi$ (and similarly for a mixed state $\Gamma$) are
\begin{equation}
   \rho_{\psi}(\mathbf{r}_1) = N \int_{\mathbb R^{3(N-1)}\times\mathbb{S}^N} \psi(\mathbf{x}_1,\ldots,\mathbf{x}_N)^{\dagger} \psi(\mathbf{x}_1,\ldots,\mathbf{x}_N) d\mathbf{x}_2\cdots d\mathbf{x}_N
\end{equation}
and
\begin{equation}
  \label{eqJmDEF}
 \begin{split}
   \jmvecix{\psi}(\mathbf{r}_1) & = N\, \mathrm{Im} \int_{\mathbb R^{3(N-1)}\times \mathbb{S}^N} \psi(\mathbf{x}_1,\ldots,\mathbf{x}_N)^{\dagger} \nabla_1 \psi(\mathbf{x}_1,\ldots,\mathbf{x}_N) d\mathbf{x}_2\cdots d\mathbf{x}_N 
          \\
    & \ \ \ + \gspin N \nabla_1 \times \int_{\mathbb R^{3(N-1)}\times\mathbb{S}^N} \psi(\mathbf{x}_1,\ldots,\mathbf{x}_N)^{\dagger} \op{\mathbf{S}}_1 \psi(\mathbf{x}_1,\ldots,\mathbf{r}_N) d\mathbf{x}_2\cdots d\mathbf{x}_N.
 \end{split}
\end{equation}
Above, we let $\mathbf{x}_k=(\mathbf{r}_k,\omega_k)$ and the integral sign denotes integration both over spatial coordinates $\mathbf{r}_k$ and over discrete spin coordinates $\omega_k \in \mathbb{S}$. The magnetization current density
$\jmvec(\mathbf{r}) = \jpvec(\mathbf{r}) + \gspin \nabla\times\mathbf{m}(\mathbf{r})$ is the sum of the paramagnetic current density
and the spin-current density. Both the magnetization current density and the paramagnetic current density are gauge-dependent quantities. The gauge-invariant, physical current density is given by $\mathbf{j}(\mathbf{r}) = \jmvec(\mathbf{r}) + \rho(\mathbf{r}) \, \mathbf{A}(\mathbf{r})$. To obtain the paramagnetic term $(\jmvec|\mathbf{A})$ above, it is necessary to transfer the curl operator from the magnetic vector potential to the spin density---that is, $(\mathbf{m}|\nabla\times\mathbf{A}) = (\nabla\times\mathbf{m}|\mathbf{A})$. As is common, we understand the kinetic energy operators $\frac{1}{2} (-i\nabla_l+\mathbf{A}(\mathbf{r}_l))^2$ and $-\frac{1}{2} \nabla_l^2$ as quadratic forms. When the spin-Zeeman term is understood as arising from the Pauli kinetic energy $\frac{1}{2} (\boldsymbol{\sigma}_l\cdot(-i\nabla_l+\mathbf{A}(\mathbf{r}_l)))^2$, this same point of view means that $\mathbf{B}(\mathbf{r}_l)=\nabla_l\times\mathbf{A}(\mathbf{r}_l)$ is a distributional derivative, to be transferred to the wave function product it is integrated with. Hence, $(\mathbf{m}|\nabla\times\mathbf{A}) = (\nabla\times\mathbf{m}|\mathbf{A})$ by definition.

Two remarks can made at this stage about Eq.~\eqref{eqEvANAIVE}: Firstly, the search domains for the different minimizations should be specified. Most of the appeal of density functional theory comes from the fact that the inner minimization is a {\it universal} functional of $(\rho,\jmvec)$. Hence, not only the expectation value itself in the inner minimization, but also the inner search domain should be free of dependencies on the external potentials. To setup a Lieb formalism, also the outer search domain needs to be independent of the external potentials.
Secondly, and even more fundamentally for a density functional theory, the expectation value of the kinetic energy has been split into a diamagnetic, a paramagnetic, and a canonical part that may not be separately finite. This is necessary to express $E(v,\mathbf{A})$ in terms of a current-density functional at all.

\subsection{Choice of function spaces}

For the ground-state problem of the Schr\"odinger--Pauli Hamiltonian,
the natural choice of wave-function space is the magnetic Sobolev space defined by $L^2$ functions with finite physical kinetic energy,
\begin{equation}
  \MSobolev{\mathbf{A}}(\mathbb{R}^{3N}\times\mathbb{S}^N)= \{\psi \in
  L^2| (-i\nabla_l + \mathbf{A}(\mathbf{r}_l))\psi \in L^2\}, \label{eq:mag-sob}
\end{equation}
where we leave implicit the restriction to properly anti-symmetric and
normalized $\psi$ (strictly speaking,
$\MSobolev{\mathbf{A}}(\mathbb{R}^{3N}\times\mathbb{S}^N)$ is then a
subset of a magnetic Sobolev space) and the index $l$ should be
understood as a generic particle index---that is, the stated condition
holds for all $1 \leq l \leq N$. 
The magnetic Sobolev space $\MSobolev{\mathbf{A}}$ is natural since
it makes the weak formulation of the ground-state problem
well defined. In particular, any eigenfunction is in
$\MSobolev{\mathbf{A}}$. On the other hand, this space depends on the magnetic vector potential. Note also that decomposition of the kinetic energy into terms such as $\bra{\psi} -\frac{1}{2} \sum_l \nabla_l^2 \ket{\psi}$, $(\jpvecix{\psi}|\mathbf{A})$, and $(\rho_{\psi}|A^2)$ is not in general possible since finite physical kinetic energy does not guarantee that these terms are separately finite.
For the purposes of constructing a CDFT, these two facts are serious obstacles. To ensure that all three terms are separately finite, we may make the additional assumption that the paramagnetic term, $(\jpvecix{\psi}|\mathbf{A})$, is finite. This suffices because the other two terms are always positive. Alternatively, it suffices to assume that the sum of the canonical kinetic energy and the diamagnetic term is finite as this in turn entails a finite paramagnetic term.

The condition $\psi \in \MSobolev{\mathbf{A}}(\mathbb{R}^{3N}\times \mathbb{S}^N)$ is a condition on the physical kinetic energy, excluding the spin-Zeeman term. To ensure the finiteness of $(\jmvec|\mathbf{A})$, we must verify that $(\nabla\times\mathbf{m}|\mathbf{A})$, in addition to $(\jpvec|\mathbf{A})$, is finite. Fortunately, this follows automatically. The Cauchy--Schwarz inequality directly yields
\begin{align}
   |(\jpvecix{\psi}|\mathbf{A})|^2 & \leq \sum_k \bra{\psi} -\nabla_k^2 \ket{\psi} \cdot (\rho|A^2), \\
   |(\nabla\times\mathbf{m}_{\psi}|\mathbf{A})|^2 & \leq \frac{3}{4} \sum_k \bra{\psi} -\nabla_k^2 \ket{\psi} \cdot (\rho|A^2),
\end{align}
since the squared spin operator for a single particle index is a multiple of the identity operator, $\op{\mathbf{S}}_1^2 = \tfrac{3}{4}$. For the spin-dependent term, we now have the following implication: $\psi \in \MSobolev{\mathbf{A}}(\mathbb{R}^{3N}\times\mathbb{S}^N)$ and finiteness of $(\jpvec|\mathbf{A})$ entails finiteness of the canonical kinetic energy and the diamagnetic term, which in turn entails finiteness of $(\nabla\times\mathbf{m}_{\psi}|\mathbf{A})$ and $(\jmvec|\mathbf{A})$.

A wave function with finite canonical kinetic energy belongs to the standard Sobolev space,
\begin{equation}
   \label{eqPSISOBOLEV}
   \psi \in \MSobolev{\mathbf{0}}(\mathbb{R}^{3N}\times\mathbb{S}^N)= \{\psi \in L^2 | \nabla_l \psi \in L^2\}.
\end{equation}
If the magnetic vector potentials are confined to the space of bounded functions, $\mathbf{A} \in \mathbf{L}^{\infty}$, no further restrictions beyond Eq.~\eqref{eqPSISOBOLEV} are needed~\cite{LAESTADIUS_IJQC114_1445}. However, in order to allow for the present CDFT formalism to later be specialized to LDFT in Sec.~\ref{secLDFT},  we shall choose a wave-function space that automatically gives rise to a finite physical kinetic energy and a finite paramagnetic term when the external vector potential grows linearly without bound such as $\mathbf{A}(\mathbf{r}) = \frac{1}{2} \mathbf{B} \times \mathbf{r}$. As a by-product, this also ensures a well-defined \magm{magnetic moment}. To achieve this, we introduce a weight function $g(\mathbf{r})$ and take the wave-function space to be defined by
\begin{equation}
   \label{eqPsiinWMSobolev}
   \psi \in \WMSobolev{\mathbf{0}}(\mathbb{R}^{3N}\times\mathbb{S}^N)= \{g(\mathbf{r}_l) \psi \in L^2 \, | \, \nabla_l \psi \in \mathbf{L}^2\}.
\end{equation}
We have in mind the specific choice
\begin{equation}
  g(\mathbf{r}) = \sqrt{1 + |\mathbf{r}|^2},
\end{equation}
which guarantees a finite paramagnetic term $(\jmvec|\mathbf{A})$ and finite \magm{magnetic moments}. However, much of the discussion remains valid for a generic $g(\mathbf{r})$ bounded away from zero to ensure that $1/g(\mathbf{r}) \in L^{\infty}(\mathbb{R}^3)$ is a bounded function. This slightly generalizes the previous work by Laestadius~\cite{LAESTADIUS_IJQC114_1445}, where the choice $g=1$ was implicit---that is, the unweighted Sobolev space was used.

The membership $g\, \jpvec \in \mathbf{L}^1$ is a straightforward application of the Cauchy--Schwarz inequality and the fact that $g^2\rho$ and the kinetic energy density are $L^1$ functions: 
\begin{equation}
 \int_{\mathbb R^3} |g\,\jpvec|d\mathbf r_1\leq \left(\int_{\mathbb R^3} g^2\rho d\mathbf r_1 \right)^{1/2}\left(N\int_{\mathbb{R}^{3N}\times\mathbb{S}^N} |\nabla_1 \psi|^2 d\mathbf{x}_1 \cdots d\mathbf{x}_N \right)^{1/2} <\infty.
\end{equation}
A similar Cauchy--Schwarz estimate shows that also the spin-current belongs to the same space, $g \nabla\times\mathbf{m} \in \mathbf{L}^1$. Hence, the total magnetization current $\jmvec$ is in the same space.

To summarize, Eq.~\eqref{eqPsiinWMSobolev} entails that the densities
belong to the function spaces
\begin{align}
     X & := \{\rho \, | \, g^2 \rho \in L^1(\mathbb R^3), \rho \in L^3(\mathbb{R}^3)\},
            \\
    Y & := \{\jmvec \, | \, g \, \jmvec  \in \mathbf L^1(\mathbb{R}^3)\}.
\end{align}
This is a slight modification of the result $\rho \in L^1 \cap L^3$
found by Lieb~\cite{LIEB_IJQC24_243} for a standard, unweighted
Sobolev space. The H\"older inequality guarantees finite interactions
$(\rho|v)$, $(\rho|A^2)$, and $(\jmvec|\mathbf{A})$ 
when the external potentials belong to the spaces
\begin{align}
   X^* & := \{ v_1+v_2 \, | \, g^{-2} v_1 \in L^{\infty}(\mathbb{R}^3), \  v_2  \in L^{3/2}(\mathbb{R}^3) \},
       \\
  Y^* & :=  \{ \mathbf{A} \, | \, g^{-1} \mathbf{A}  \in \mathbf{L}^{\infty}(\mathbb{R}^3)\}.
\label{PotSp}
\end{align}
Indeed, $Y^*$ is the continuous dual of the weighted Lebesgue space
$Y$, whereas $X^*$ is the continuous dual space of $X$,
a Banach spaces under the norm $\|\rho\|_X := \|g^2\rho\|_{L^1} +
\|\rho\|_{L^3}$.
With this choice, the density-dependent energy term $(\rho|v+\tfrac{1}{2}A^2)$ is guaranteed to be finite for all $\rho \in X$, $v \in X^*$, and $\mathbf{A} \in Y^*$. As a consequence, the diamagnetic potential must also belong to the dual space of the densities---that is, $A^2 \in X^*$ and $v + \tfrac{1}{2} A^2 \in X^*$.

We finally remark that finiteness of the canonical kinetic energy, Eq.~\eqref{eqPsiinWMSobolev}, implies finiteness of the $\jpvec$-corrected von Weizs\"acker energy~\cite{BATES_JCP137_164105,LAESTADIUS_IJQC114_1445}:
\begin{equation}
  \int \frac{|\nabla\rho_{\psi}|^2 + 4|\jpvecix{\psi}|^2}{8\rho_{\psi}} d\mathbf{r} \leq \bra{\psi} -\tfrac{1}{2} \nabla^2 \ket{\psi} < +\infty.
\end{equation}
Recent work has shown this bound to be valid also when $\jpvecix{\psi}$ is replaced by $\jmvecix{\psi}$~\cite{TELLGREN_MANUSCRIPT}.

\subsection{Constrained search and Lieb formulations of CDFT}

Having discussed the search domains, we introduce the notation $\op{H}_0 = \op{H}(0,\mathbf{0})$ and define the universal functionals
\begin{align}
  \Fvr(\rho,\jmvec) & = \inf_{\substack{\psi \in \WMSobolev{\mathbf{0}} \\ \psi \mapsto \rho, \jmvec}} \bra{\psi} \op{H}_0 \ket{\psi},    \\
  \Fvrdm(\rho,\jmvec) & = \inf_{\substack{\Gamma \in \DMset \\ \Gamma \mapsto \rho, \jmvec}} \trace{ \Gamma \op{H}_0}.
\end{align}
Within the present, mildly restricted choice of search domains, the first functional is simply the pure-state constrained-search functional of Vignale and Rasolt~\cite{VIGNALE_PRL59_2360}, while the second functional is its extension to mixed states, with $\DMset$ denoting the convex set of all valid density operators~\cite{TELLGREN_PRA89_012515}. Because all pure states $\psi$ can be represented by density operators $\Gamma = \ket{\psi} \bra{\psi}$, it follows immediately that $\Fvrdm(\rho,\jmvec) \leq \Fvr(\rho,\jmvec)$.
The pure- and mixed-state functionals can be used interchangeably in the CDFT variation principle,
\begin{equation}
 \begin{split}
  \label{eqCDFTVarPrinc}
	E(v,\mathbf{A}) = \inf_{\psi \in \WMSobolev{\mathbf{0}}} \bra{\psi} H(v,\mathbf{A}) \ket{\psi} &= \inf_{(\rho,\jmvec)\in X\times Y}
	\Big( (\rho|v + \tfrac{1}{2} A^2) + (\jmvec|\mathbf{A}) +  \inf_{\substack{\psi \in \WMSobolev{\mathbf{0}} \\ \psi \mapsto \rho, \jmvec}} \bra{\psi} \op{T} + \op{W} \ket{\psi}\Big)
      \\
	& = \inf_{(\rho,\jmvec)\in X\times Y}\Big( (\rho|v + \tfrac{1}{2} A^2) + (\jmvec|\mathbf{A}) +  \Fvr(\rho,\jmvec)\Big)
      \\
	& = \inf_{(\rho,\jmvec)\in X\times Y}\Big( (\rho|v + \tfrac{1}{2} A^2) + (\jmvec|\mathbf{A}) +  \Fvrdm(\rho,\jmvec)\Big).
 \end{split}
\end{equation} 
The last equality follows because $E(v,\mathbf{A})$ could just as well have been expressed as an infimum over mixed states.

Because our choice of wave-function space amounts to a mild restriction to wave functions with well-defined second-order moments, 
there is a risk that some potentials $(v',\mathbf{A}') \in X^*\times Y^*$ have ground states that decay so slowly 
that $\psi' \notin \WMSobolev{\mathbf{0}}$ and $(\rho',\jmvec') \notin X\times Y$. Such ground states could have an infinite diamagnetic energy, 
which would be compensated by an infinite, negative paramagnetic energy. For the present purposes, we simply accept the fact that minima
attained by slowly decaying ground states without well-defined second-order moments may be lost in this formulation. 
However, the restrictions on the function space do not affect the value of the infimum defining $E(v,\mathbf{A})$, it only means that there 
are more instances where the infimum is not a minimum, attained by a wave function---for a careful discussion of this aspects, we refer to Kvaal et al.~\cite{KVAAL_MANUSCRIPT}. We expect molecular systems to feature exponentially decaying ground states, so this mild restriction leaves many areas of application unaffected. 
Several rigorous results exist that establish, for example, the exponential decay of ground-state wave functions for wide classes of potentials in the absence of magnetic fields~\cite{SIMON_JMP41_3523} and in the presence of constant or asymptotically constant fields~\cite{AVRON_CMP79_529,SORDONI_CPDE23_247}.

We now turn to the convexity properties of the functionals. The pure-state constrained.search functional $\Fvr(\rho,\jmvec)$ is not convex---see Proposition 8 in 
Ref.~\onlinecite{LAESTADIUS_IJQC114_1445}, where we here in addition assume that the ground states in the counterexample are elements of $\WMSobolev{\mathbf{0}}$. By contrast, the mixed-state functional $\Fvrdm(\rho,\jmvec)$ is jointly convex in $(\rho,\jmvec)$ because the map $\Gamma \mapsto (\rho,\jmvec)$ is linear. If the ground-state energy functional $E(v,\mathbf{A})$ had featured a universal functional and potentials paired {\it linearly} with densities, it would have had the form of a Legendre--Fenchel transform and be manifestly concave. However, the diamagnetic term $\frac{1}{2} (\rho|A^2)$ prevents such an identification and the existence of diamagnetic molecules indeed shows that $E(v,\mathbf{A})$ is not concave in $\mathbf A$. The non-concavity
can be remedied by a change of variables that absorbs the diamagnetic term into the scalar potential~\cite{TELLGREN_PRA89_012515},
\begin{equation}
   u := v + \frac{1}{2} A^2.
\end{equation}
Since $u \in X^*$ follows from $v \in X^*$ and $\mathbf{A} \in Y^*$, this change of variables ``stays within'' the already specified space of scalar potentials. Defining $\bar{H}(u,\mathbf{A}) := H(u - \tfrac{1}{2} A^2, \mathbf{A})$, we may now write
\begin{equation}
 \begin{split}
  \label{eqCDFTVarPrinc2}
	\bar{E}(u,\mathbf{A}) = \inf_{\Gamma \in \DMset} \trace{\Gamma \bar{H}(u,\mathbf{A})} &= \inf_{(\rho,\jmvec)\in X\times Y}\Big( (\rho|u) + (\jmvec|\mathbf{A}) +  \Fvrdm(\rho,\jmvec)\Big).
 \end{split}
\end{equation} 
Note that $\bar{E}(u,\mathbf{A})= E(u - \tfrac{1}{2} A^2,\mathbf{A})$. It can be remarked that not all potential pairs in $X^*\times Y^*$ correspond to physical systems. For example, some pairs correspond to harmonic oscillator-type potentials with negative sign. A concrete example is a one-electron system subject to $u(\mathbf{r}) = -1/r$ and $\mathbf{A}(\mathbf{r}) = \tfrac{1}{2} \mathbf{B}\times\mathbf{r}$; in this case the Hamiltonian $\bar{H}(u,\mathbf{A})$ fails to be bounded from below as there is no diamagnetic term to balance the orbital Zeeman effect~\cite{SAVIN_MP115_13}.

The functional $\bar{E}(u,\mathbf{A})$ has the form of a Legendre--Fenchel transform and is jointly concave in $(u,\mathbf{A})$. 
Repeated Legendre--Fenchel transformation now yields the conjugate pair
\begin{equation}
  \label{eqFLFTransf}
 \begin{split}
	 \bar{F}(\rho,\jmvec) = \sup_{(u,\mathbf{A}) \in X^*\times Y^*} \big( \bar{E}(u,\mathbf{A}) - (\rho|u) - (\jmvec|\mathbf{A}) \big)
 \end{split}
\end{equation} 
and
\begin{equation}
 \begin{split}
  \label{eqCDFTVarPrinc3}
	\bar{E}(u,\mathbf{A}) = \inf_{(\rho,\jmvec)\in X\times Y}\Big( (\rho|u) + (\jmvec|\mathbf{A}) +  \bar{F}(\rho,\jmvec)\Big).
 \end{split}
\end{equation}
Any function given as a Legendre--Fenchel transform automatically has
some regularity. Thus, $\bar{F}$ is convex and  lower
semicontinuous (l.s.c.), while $\bar{E}$ is concave and upper
semicontinuous (u.s.c.); see, for example, Ref.~\cite{VANTIEL}.  Lower semicontinuity of 
$f(x)$ means that, if $x_k \to x$,
then $\lim \inf_{k\rightarrow \infty} f(x_k) \geq f(x)$ (and similarly for u.s.c. functions).

All three universal functionals $\Fvr$, $\Fvrdm$, and $\bar{F}$ are equivalent in the sense that they can be used interchangeably in the CDFT variation principle in Eqs.~\eqref{eqCDFTVarPrinc}, \eqref{eqCDFTVarPrinc2}, and \eqref{eqCDFTVarPrinc3}, producing the
right ground-state energy. In the language of Ref.~\cite{KVAAL_JCP143_184106} they are {\it admissible functionals}. However, $\bar{F}$ is the \emph{unique} admissible functional that is both convex  and l.s.c.---that is, it is the unique functional that satisfies the variation principle~\eqref{eqFLFTransf}. This is significant, because this functional it the only one that can---in principle, at least---be computed from $\bar{E}(u,\mathbf{A})$.

Removing the minimization in Eq.~\eqref{eqCDFTVarPrinc2}
and moving the density-potential pairing integrals to the left-hand side, we obtain
\begin{equation}
  \bar{E}(u,\mathbf{A}) - (\rho|u) - (\jmvec|\mathbf{A}) \leq \Fvrdm(\rho,\jmvec),
\end{equation}
for all $(\rho,\jmvec) \in X\times Y$ and all $(u,\mathbf{A}) \in X^*\times Y^*$. Taking the supremum over potentials,
\begin{equation}
    \label{eqFlFvrBound}
    \bar{F}(\rho,\jmvec) \leq \Fvrdm(\rho,\jmvec).
\end{equation}
The above argument is not specific to $\Fvrdm$, but valid for any
admissible functional. Therefore, in
general, $\bar{F}\leq F$ for any admissible $F$, such as $\Fvr$. 
For the l.s.c.\ of $\bar{F}$, in particular, see the straightforward argument adapted to the
CDFT setting in the proof of Proposition 12 in
Ref.~\onlinecite{LAESTADIUS_IJQC114_1445}.

The non-convexity of the pure-state functional means that,
in general, $\bar{F}(\rho,\jmvec) \neq \Fvr(\rho,\jmvec)$. The
identification of $\bar{F}$ and $\Fvrdm$ hinges on the l.s.c.\ of the
latter functional, which is presently an open question; see
Ref.~\onlinecite{KVAAL_MANUSCRIPT}. If $\Fvrdm$ is l.s.c., it follows
that $\Fvrdm = \bar{F}$ since $\bar{F}$ is unique. It is remarkable
that, in standard DFT, we have the result $\bar{F}(\rho) =
F_\text{DM}(\rho)$~\cite{LIEB_IJQC24_243}. 

\section{LDFT: A DFT formalism for affine vector potentials}
\label{secLDFT}

We study in this section how the paramagnetic CDFT of Vignale and Rasolt can be specialized and simplified when the vector-potential space is suitably constrained. Let us restrict attention to magnetic vector potentials that can be Taylor expanded around some reference point $\mathbf{G}$,
\begin{equation}
  A_k(\mathbf{r}) =  A^{(0)}_k + \sum_{l=1}^3 A^{(1)}_{kl} (r_l - G_l) + \sum_{l,m=1}^3 A^{(2)}_{klm} (r_l - G_l) (r_m - G_m) + \ldots
\end{equation}
Truncating at some fixed order $n$ is equivalent to restricting the function space of allowed vector potentials to polynomials of degree $n$ in $(x,y,z)$. 
In what follows, we study the case $n=1$ with the additional gauge condition $\nabla\cdot\mathbf{A}(\mathbf{r}) = 0$. 
However, everything can be straightforwardly generalized to an arbitrary finite order $n$ without additional gauge constraints.

We thus focus here on vector potentials of the form $\mathbf{A} = \tfrac{1}{2} \mathbf{B}\times\mathbf{r}$. We find it instructive to not completely eliminate all gauge degrees of freedom, allowing constant shifts $\mathbf{a}$ of the vector potentials. Retention of these gauge degrees of freedom makes it easier to discuss, for example,  additivity of physical energies over independent subsystems. Hence, we take the scalar potential space to be $X^*$ (as in the previous section) and the vector potential space to be a six-dimensional vector space,
\begin{align}
  \mathbf{A} & \in \ALinSpace := \{ \mathbf{a} + \tfrac{1}{2} \mathbf{B}\times(\mathbf{r}-\mathbf{G}) | \mathbf{a} \in \mathbb{R}^3, \mathbf{B} \in \mathbb{R}^3\} \subset Y^*,
  \label{eqALinSpace}
\end{align}
where the weight function is $g(\mathbf{r}) = \sqrt{1+|\mathbf{r}|^2}$.  This choice of function space guarantees that both 
the canonical momentum $\mathbf{p} = \int \jpvec d\mathbf{r}$ and the \magm{paramagnetic moment} $\mathbf{L}_{\mathbf{C}} = \int (\mathbf{r}-\mathbf{C})\times\jmvec d\mathbf{r}$ are finite. It is also worth being explicit about the role of spin:
\begin{align}
   \mathbf{p} & = \int \jpvec d\mathbf{r} = \int \jmvec d\mathbf{r},
              \\
   \mathbf{L}_{\mathbf{C}} & = \int (\mathbf{r}-\mathbf{C})\times\jmvec d\mathbf{r} = \int (\mathbf{r}-\mathbf{C})\times\jpvec d\mathbf{r} + \gspin \mathbf{S}.
\end{align}
Hence, if $\gspin \neq 0$, then the spin degrees of freedom contribute to the \magm{magnetic moment} but not to the linear momentum. 
The physical momentum can similarly be expressed as $\boldsymbol{\pi} = \int \mathbf{j} d\mathbf{r} = \mathbf{p} + \int \rho \mathbf{A} d\mathbf{r}$ and the \magm{physical magnetic moment} as $\mathbf{J}_{\mathbf{C}} = \int (\mathbf{r}-\mathbf{C})\times \mathbf{j} d\mathbf{r} = \mathbf{L}_{\mathbf{C}} + \int (\mathbf{r}-\mathbf{C})\times \rho \mathbf{A} d\mathbf{r}$. Magnetic moments relative to different reference points are related as $\mathbf{L}_{\mathbf{C}} = \mathbf{L}_{\mathbf{D}}-(\mathbf{D}-\mathbf{C})\times\mathbf{p}$ and $\mathbf{J}_{\mathbf{C}} = \mathbf{J}_{\mathbf{D}}-(\mathbf{D}-\mathbf{C})\times\boldsymbol{\pi}$. Moreover, 
for any other eigenstate of the Hamiltonian (and, in particular, the ground state), we have 
\begin{equation}
  \boldsymbol{\pi} = \mathbf{p} + N \mathbf{a} + \frac{1}{2} \mathbf{B} \times \boldsymbol{\mu}_{\mathbf{G}} = \mathbf{0},
\end{equation}
where $\boldsymbol{\mu}_{\mathbf{G}} = \int (\mathbf{r}-\mathbf{G}) \, \rho d\mathbf{r}$ is the electric dipole moment relative to $\mathbf{G}$. As a consequence, for energy eigenstates, the \magm{physical magnetic
  moment} $\mathbf{J}_{\mathbf{C}}$ is independent of the reference
point $\mathbf{C}$.

The only allowed gauge transformations are now {\it gauge shifts}---that is, constant shifts of the magnetic vector potential. The canonical momentum and \magm{paramagnetic moment} transform in a simple manner under such gauge shifts,
\begin{align}
  \psi & \mapsto \psi e^{i \mathbf{a}'\cdot\mathbf{r}},  \\
  \mathbf{A} & \mapsto \mathbf{A} + \mathbf{a}',   \\
  \mathbf{p} & \mapsto \mathbf{p} - N \mathbf{a}',   \\
  \mathbf{L}_{\mathbf{G}} & \mapsto \mathbf{L}_{\mathbf{G}} - \boldsymbol{\mu}_{\mathbf{G}} \times \mathbf{a}',
\end{align}
where $N = \int \rho d\mathbf{r}$ is the number of electrons. We also define a gauge-shift invariant quantity that we term {\it the intrinsic magnetic moment},
\begin{equation}
 \label{eqLAMBDADEF}
  \boldsymbol{\Lambda} = \mathbf{L}_{\mathbf{G}} - \frac{\boldsymbol{\mu}_{\mathbf{G}} \times \mathbf{p}}{N}.
\end{equation}
The \magm{intrinsic magnetic moment} is identical to the \magm{paramagnetic moment} with respect to the average position of the electrons or centre of charge, $\boldsymbol{\Lambda} = \mathbf{L}_{\mathbf{R}}$ with $\mathbf{R} = \mathbf{G} + N^{-1} \boldsymbol{\mu}_{\mathbf{G}}$. In some sense, the \magm{intrinsic magnetic moment} extracts the part of $\mathbf{L}_{\mathbf{G}}$ that is independent of the physically arbitrary reference point $\mathbf{G}$. It is therefore better suited as a parameter in a functional describing a physical energy, which must be independent of arbitrary reference points. We return to this point in Sec.~\ref{secKS}, 
where we discuss the Kohn--Sham decomposition of the total energy.

\subsection{Constrained-search and Lieb formulations of LDFT}

Let us now define a ground-state functional on the space $X^* \times \ALinSpace$ by
\begin{equation}
 \begin{split}
  E(v,\mathbf{a},\mathbf{B}) & = \inf_{\psi \in \WMSobolev{\mathbf{0}}} \bra{\psi} \sum_{k=1}^N \frac{1}{2} \big(-i\nabla_k + \mathbf{A}(\mathbf{r}_k) \big)^2 + \sum_{k=1}^N v(\mathbf{r}_k) + \op{W} \ket{\psi} \\
    & = \inf_{(\rho, \mathbf{p}, \mathbf{L}_{\mathbf{G}})\in X \times \mathbb{R}^6} \big( (\rho|v+\tfrac{1}{2} A^2) + \mathbf{a}\cdot\mathbf{p} + \frac{1}{2} \mathbf{B}\cdot\mathbf{L}_{\mathbf{G}} + \inf_{\substack{\psi \in \WMSobolev{\mathbf{0}} \\ \psi \mapsto \rho, \mathbf{p}, \mathbf{L}_{\mathbf{G}}}} \bra{\psi} \op{H}_0 \ket{\psi} \big)
     \\
  &  = \inf_{(\rho, \mathbf{p}, \mathbf{L}_{\mathbf{G}}) \in X\times\mathbb{R}^6} \big( (\rho|v+\tfrac{1}{2} A^2) + \mathbf{a}\cdot\mathbf{p} + \frac{1}{2} \mathbf{B}\cdot\mathbf{L}_{\mathbf{G}} +\Fldft(\rho,\mathbf{p},\mathbf{L}_{\mathbf{G}}) \big)
     \\
  &  = \inf_{(\rho, \mathbf{p}, \mathbf{L}_{\mathbf{G}}) \in X\times\mathbb{R}^6} \big( (\rho|v+\tfrac{1}{2} A^2) + \mathbf{a}\cdot\mathbf{p} + \frac{1}{2} \mathbf{B}\cdot\mathbf{L}_{\mathbf{G}} +\Fldftdm(\rho,\mathbf{p},\mathbf{L}_{\mathbf{G}}) \big),
 \end{split}
\end{equation}
where the constrained-search functional can be defined either with pure states or with mixed states,
\begin{align}
   \Fldft(\rho,\mathbf{p},\mathbf{L}_{\mathbf{G}}) & := \inf_{\substack{\psi \in \WMSobolev{\mathbf{0}} \\ \psi \mapsto \rho, \mathbf{p}, \mathbf{L}_{\mathbf{G}}}} \bra{\psi} \op{H}_0 \ket{\psi},
           \\
   \Fldftdm(\rho,\mathbf{p},\mathbf{L}_{\mathbf{G}}) & := \inf_{\substack{\Gamma \in \DMset  \\ \Gamma \mapsto \rho, \mathbf{p}, \mathbf{L}_{\mathbf{G}}}} \trace{ \Gamma \op{H}_0 }.
\end{align}
For a magnetization current density consistent with the canonical momentum and \magm{magnetic moment}---that is, $\mathbf{p} = \int \jmvec d\mathbf{r}$ and $\mathbf{L}_{\mathbf{G}} = \int \mathbf{r}_{\mathbf{G}}\times\jmvec d\mathbf{r}$---we have the inequalities
\begin{align}
   \Fldft(\rho,\mathbf{p},\mathbf{L}_{\mathbf{G}}) & \leq \Fvr(\rho,\jmvec),
           \\
   \Fldftdm(\rho,\mathbf{p},\mathbf{L}_{\mathbf{G}}) & \leq \Fvrdm(\rho,\jmvec).
\end{align}

Next, noting that $A^2 \in X^*$, we may introduce a new potential $u := v + \tfrac{1}{2} A^2 \in X^*$ that remains in the same function space as $v \in X^*$. After this reparametrization, we obtain
\begin{equation}
  \bar{E}(u,\mathbf{a},\mathbf{B}) = \inf_{(\rho, \mathbf{p}, \mathbf{L}_{\mathbf{G}})\in X\times\mathbb{R}^6} \big( (\rho|u) + \mathbf{a}\cdot\mathbf{p} + \tfrac{1}{2} \mathbf{B}\cdot\mathbf{L}_{\mathbf{G}} + \Fldftdm(\rho,\mathbf{p},\mathbf{L}_{\mathbf{G}}) \big).
\end{equation}
This expression takes the form of a Legendre--Fenchel transformation (concave conjugation). Hence, a Lieb formalism is obtained for affine vector potentials. 
Repeated conjugation gives the conjugate pair
\begin{equation}
  \bar{F}(\rho,\mathbf{p},\mathbf{L}_{\mathbf{G}}) = \sup_{(u,\mathbf{a},\mathbf{B}) \in Y\times \mathbb{R}^6} \big( \bar{E}(u,\mathbf{a},\mathbf{B}) - (\rho|u) - \mathbf{a}\cdot\mathbf{p} - \tfrac{1}{2} \mathbf{B}\cdot\mathbf{L}_{\mathbf{G}} \big)
\end{equation}
and
\begin{equation}
  \bar{E}(u,\mathbf{a},\mathbf{B}) = \inf_{(\rho, \mathbf{p}, \mathbf{L}_{\mathbf{G}})\in X\times\mathbb{R}^6} \big( (\rho|u) + \mathbf{a}\cdot\mathbf{p} + \tfrac{1}{2} \mathbf{B}\cdot\mathbf{L}_{\mathbf{G}} + \bar{F}(\rho,\mathbf{p},\mathbf{L}_{\mathbf{G}}) \big).
\end{equation}
These two transformations constitute a generalization of Lieb's formulation of DFT to include uniform magnetic fields.

\subsection{Expectation-valuedness of the pure state functional}

In the terminology of Kvaal and Helgaker, a DFT functional is said to be expectation-valued if its finite function values are always attained by expectation value $\bra{\psi} \op{H}_0 \ket{\psi}$ for some state $\psi$~\cite{KVAAL_JCP143_184106}.
The conventional DFT functionals are expectation valued~\cite{LIEB_IJQC24_243} as is the pure-state CDFT functional $\Fvr$~\cite{LAESTADIUS_IJQC114_1445}.
It is an open question whether or not the universal CDFT functional $\Fvrdm$ is expectation valued.
Below we modify the CDFT argument in Ref.~\onlinecite{LAESTADIUS_IJQC114_1445} to prove that also the pure state LDFT functional $\Fldft$ has this property.

We first introduce the set of density triples corresponding to $N$ electrons,
\begin{equation}
\mathcal{J}_N = \{ (\rho,\mathbf{p},\mathbf{L}_{\mathbf{G}}) \in X \times \mathbb{R}^6 \ | \ \rho(\mathbf{r}) \geq 0, \ \int \rho \, d\mathbf{r} =N \}.
\end{equation}
All densities that are pure-state $N$-representable (with finite kinetic energy) are contained in this set. The LDFT functional attains its minimum whenever it is finite. In the
language of Ref.~\onlinecite{KVAAL_JCP143_184106}, $\Fldft$ is expectation valued:
\begin{theorem}
         For all triples $(\rho,\mathbf{p} ,\mathbf{L}_{\mathbf{G}})
         \in \mathcal{J}_N$ such that
         $\Fldft(\rho,\mathbf{p},\mathbf{L}_{\mathbf{G}}) < +\infty$, there exists a $\psi\in \WMSobolev{\mathbf{0}}$ such that $\psi\mapsto (\rho,\mathbf{p},\mathbf{L}_{\mathbf{G}})$ and 
     \begin{equation}
        \Fldft(\rho,\mathbf{p},\mathbf{L}_{\mathbf{G}}) = \bra{\psi} \op{H}_0 \ket{\psi}.
      \end{equation}
\end{theorem}
\begin{proof}
        Let $\{ \psi_k \} \subset \WMSobolev{\mathbf{0}}$ be a minimizing sequence such that 
        \begin{equation}
        \lim_{k\to \infty}  \bra{\psi_k} \op{H}_0 \ket{\psi_k} =  \Fldft(\rho,\mathbf p,\mathbf{L}_{\mathbf{G}}),
        \end{equation}
        and $\psi_k \mapsto (\rho,\mathbf{p},\mathbf{L}_{\mathbf{G}})$ for each $k$. This means that, for any $k$, 
we have $\int \jmvecix{k} d\mathbf{r} = \mathbf{p}$ and $\int \mathbf{r}_{\mathbf{G}} \times \jmvecix{k} d\mathbf{r} = \mathbf{L}_{\mathbf{G}}$, where 
        $\jmvecix{k} := \jmvecix{\psi_k}$. Furthermore, denoting the volume element of the combined coordinate and spin space $\mathbb{R}^{3N} \times \mathbb{S}^N$ by $d \tau$, we have
        \begin{align}
                 \label{eqExpValKinBound}
                &\sum_{l=1}^N \int_{\mathbb R^{3N} \times \mathbb{S}^N} \vert \nabla_l \psi_k \vert^2 d\tau \leq 2 \bra{\psi_k} \op{H}_0 \ket{\psi_k} \leq C,\\
                &N \int_{\mathbb{R}^{3N} \times \mathbb{S}^N} \vert \psi_k \vert^2 d\tau = \int_{\mathbb{R}^3} \rho \, d\mathbf{r} =N, \\
                &N\int_{\mathbb{R}^{3N} \times \mathbb{S}^N} |g(\mathbf r_l)\psi_k|^2 d\tau = \int_{\mathbb{R}^3} g^2 \, \rho \, d\mathbf{r} <+\infty,
        \end{align}
        which means that $\{\psi_k\}$ is bounded in $\mathcal H_0^1$. By the Banach-Alaoglu theorem $\psi_k\rightharpoonup \psi$ weakly in $\mathcal H_0^1$ for some $\psi\in\mathcal H_0^1$. The argument in the proof of Theorem 3.3 in Lieb's work~\cite{LIEB_IJQC24_243} shows that $\psi_k\to\psi$ in $L^2$-norm and $\psi\mapsto\rho$. We may then conclude that $\psi\in \WMSobolev{\mathbf{0}}$, because
        \begin{equation}
        N\int_{\mathbb{R}^{3N} \times \mathbb{S}^N} g^2|\psi|^2 \, d\tau = \int_{\mathbb{R}^3} g^2 \rho_{\psi} \, d\mathbf{r} = \int_{\mathbb{R}^3} g^2\rho \, d\mathbf{r} <+\infty.
        \end{equation}
        Furthermore, since $\psi_k\rightharpoonup\psi$ in $\mathcal
        H_0^1$ and the expectation value of $\op{H}_0$ is weakly l.s.c., it holds that
        \begin{equation}
            \label{eqLIMkFL}
        \Fldft(\rho,\mathbf p,\mathbf{L}_{\mathbf{G}}) = \lim_{k\to\infty} \bra{\psi_k} \op{H}_0 \ket{\psi_k} \geq \bra{\psi} \op{H}_0 \ket{\psi}.
        \end{equation}        
        If we can  prove $\jmvecix{k} \to \jmvecix{\psi}$ as well as $r_l (\jmvecix{k})_{l'} \to r_l(\jmvec{\psi})_{l'}$ in $\mathbf L^1$-norm, we have
        \begin{align}
          \label{eqExpValJintLim}
        &\int \jmvecix{\psi} \, d\mathbf{r} = \lim_{k\to\infty} \int \jmvecix{k} \, d\mathbf{r} = \mathbf{p}, \\
        & \int \mathbf{r}_{\mathbf{G}} \times \jmvecix{\psi} \, d\mathbf{r} = \lim_{k\to\infty} \int \mathbf{r}_{\mathbf{G}} \times \jmvecix{k} \, d\mathbf{r} = \mathbf{L}_{\mathbf{G}},
          \label{eqExpValJintLimX}
        \end{align}
        and the proof would be complete.

Equation~\eqref{eqExpValJintLim} implies a simple rule for changing the reference point from $\mathbf{G}$ to another location in Eq.~\eqref{eqExpValJintLimX}.
Hence, without loss of generality, we simplify the notation by moving the reference point to the origin in what follows, 
writing $\mathbf{r}$ instead of $\mathbf{r}_{\mathbf{G}}$. 
To prove convergence of the current-density sequence, we note that $\phi_k := g\psi_k$ is a bounded sequence in $L^2$. 
Again, by the Banach--Alaoglu theorem, $\phi_k\rightharpoonup \phi$ weakly in $L^2$ for some $\phi\in L^2$. The argument used in the proof of Theorem 3.3 in 
Ref.~\onlinecite{LIEB_IJQC24_243} shows that (for a subsequence) $\phi_k\to \phi$ in $L^2$-norm. 
This argument relies on the fact that weak convergence $\phi_k \rightharpoonup \phi$ is given and it only remains to show 
norm convergence $\|\phi_k\|_{L^2} \to \|\phi\|_{L^2}$ to prove that $\phi_k \to \phi$ in $L^2$. 
Since $\lim_{k\to\infty} \|\phi_k\|_{L^2} \geq \|\phi\|_{L^2}$ is a generic consequence of weak convergence, only the reverse inequality needs to be proven below. 

Let $\varepsilon >0$ be given. Since $g^2\rho\in L^1$, we may choose a characteristic function $\chi$ of some bounded set in $\mathbb R^3$ such that 
        \begin{equation}
        \varepsilon> \int_{\mathbb R^3} g^2\rho(1-\chi) d\mathbf{r} = \int_{\mathbb R^{3N} \times \mathbb{S}^N} |\phi_k|^2\sum_{j=1}^N (1-\chi(\mathbf r_j)) d\tau.
        \end{equation}
        Note that $S:=\Pi_{j=1}^N \chi(\mathbf r_j)$ satisfies $\sum_{j=1}^N (1-\chi(\mathbf r_j)) \geq 1 - S$. Since $S\phi_k \to S\phi$ in the $L^2$-norm (see Theorem 3.3 in Ref.~\onlinecite{LIEB_IJQC24_243}), we obtain
        \begin{equation}
        \varepsilon \geq \lim_k \int_{\mathbb R^{3N} \times \mathbb{S}^N} |\phi_k|^2(1-S) d\tau
        = \lim_k \int_{\mathbb R^{3N}\times \mathbb{S}^N} |\phi_k|^2 d\tau - \int_{\mathbb R^{3N}\times \mathbb{S}^N} S|\phi|^2 d\tau,
        \end{equation}
giving
        \begin{equation}
        \int_{\mathbb R^{3N}\times \mathbb{S}^N} |\phi|^2 d\tau + \varepsilon\geq \int_{\mathbb R^{3N}\times \mathbb{S}^N} S|\phi|^2 d\tau + 
        \varepsilon \geq \lim_k \int_{\mathbb R^{3N}\times \mathbb{S}^N} |\phi_k|^2 d\tau.
        \end{equation}
        
Next, let $\delta >0$ be given and choose $R>0$ such that $\int_{B_R^c} |\phi - g\psi|^2 d\tau \leq \delta^2$, where we leave the spin degrees of freedom implicit.
Let $\mathcal M \subset B_R$ be any measurable set  and note that $\int_{\mathcal M}\phi d\tau = \lim_k \int_{\mathcal M} g\psi_k d\tau = \int_{\mathcal M} g\psi d\tau$, implying
that $\phi = g\psi$ almost everywhere on $B_R$. Choosing $N$ such that $k>N$ implies $\Vert \phi_k-\phi \Vert_{L^2}\leq \varepsilon$, we obtain 
        \begin{equation}
        \Vert \phi_k - g\psi \Vert_{L^2}\leq \Vert \phi_k - \phi \Vert_{L^2}
        + \Vert \phi - g\psi \Vert_{L^2} \leq \delta +\left(\int_{B_R^c}| \phi - g\psi|^2 d\tau \right)^{1/2} \leq 2\delta.   
        \end{equation}

        We now show that, with $l,l'\in \{1,2,3\}$ and a bounded self-adjoint operator $\op{\Omega}$ with operator norm less than 1,
        \begin{equation}
                \label{g2}
           \begin{split}
                \int_{\mathbb R^{3N}\times \mathbb{S}^N} r_l  (\op{\Omega} \psi_k)^* \partial_{l'} \psi_k d\tau  & =
                \int_{\mathbb R^{3N}\times \mathbb{S}^N} r_l ( (\op{\Omega}\psi_k - \op{\Omega}\psi)^* \partial_{l'} \psi_k 
                + (\op{\Omega} \psi)^* \partial_{l'} (\psi_k-\psi) + (\op{\Omega} \psi)^* \partial_{l'} \psi) d\tau 
                       \\
                & \to \int_{\mathbb R^{3N}\times \mathbb{S}^N} r_l  (\op{\Omega} \psi)^* \partial_{l'} \psi d\tau
            \end{split}
        \end{equation}
        as $k\to\infty$. The difference between the indicated limit and the left hand-side vanishes because, using $\Vert \phi_k - g\psi \Vert_{L^2}\to 0$, we have
        \begin{equation}
          \begin{split}
        \Big\vert\int_{\mathbb R^{3N}\times\mathbb{S}^N} r_l (\op{\Omega} \psi_k- \op{\Omega} \psi)^* \partial_{l'} \psi_k d\tau \Big\vert &  \leq \left( \int_{\mathbb R^{3N}\times\mathbb{S}^N} r_l^2 |\op{\Omega} \psi_k- \op{\Omega} \psi|^2 d\tau  \right)^{1/2} \left( \int_{\mathbb R^{3N}\times\mathbb{S}^N} |\partial_{l'} \psi |^2  d\tau \right)^{1/2} \\
        & \leq \Vert \phi_k - g\psi \Vert_{L^2}
        \Vert\nabla \psi \Vert_{L^2} \to 0 
          \end{split}
        \end{equation}
        and furthermore since $r_l \op{\Omega} \psi^*\in L^2$, 
        \begin{equation}
        \int_{\mathbb R^{3N}\times\mathbb{S}^N} r_l (\op{\Omega} \psi)^* \partial_{l'}(\psi_k-\psi) d\tau \to 0.
        \end{equation}
        
        Having established Eq.~\eqref{g2}, we can set $\op{\Omega}$ to the identity operator to conclude that 
        $\int_{\mathbb R^{3}} r_l (\jpvecix{k})_{l'} d\mathbf{r} \to \int_{\mathbb R^{3}} r_l (\jpvecix{\psi})_{l'} d\mathbf{r}$. We can also set $\op{\Omega}$ to be any component of the spin operator to conclude that the spin-current contribution in Eq.~\eqref{eqJmDEF} converges in the same way. Hence, we have $\int_{\mathbb R^{3}} r_l (\jmvecix{k})_{l'} d\mathbf{r} \to \int_{\mathbb R^{3}} r_l (\jmvecix{\psi})_{l'} d\mathbf{r}$. A similar argument shows that $\int_{\mathbb R^{3}} (\jmvecix{k})_{l} d\mathbf{r} \to \int_{\mathbb R^{3}}  (\jmvecix{\psi})_{l} d\mathbf{r}$. 
\end{proof}

An immediate generalization of the above theorem is obtained by noting that the proof remains valid when the electron--electron repulsion operator is scaled by a non-negative parameter $\lambda$. The only parts affected by the scaling are Eqs.~\eqref{eqExpValKinBound} and \eqref{eqLIMkFL}, which still hold for any non-negative scaling.
\begin{theorem}
      \label{thmGenExpVal}
      Consider the modified potential-free Hamiltonian
\begin{equation}
     \op{H}_0^{\lambda} = \op{T} + \lambda \op{W}.
\end{equation}
For all $\lambda \geq 0$ and all triples $(\rho,\mathbf{p} ,\mathbf{L}_{\mathbf{G}}) \in \mathcal{J}_N$, there exists a $\psi_{\lambda} \in \WMSobolev{\mathbf{0}}$ such that $\psi_{\lambda} \mapsto (\rho,\mathbf{p},\mathbf{L}_{\mathbf{G}})$ and 
        \begin{equation}
        \Fldft^{\lambda}(\rho,\mathbf{p},\mathbf{L}_{\mathbf{G}}) := \inf_{\psi' \mapsto \rho,\mathbf{p},\mathbf{L}_{\mathbf{G}}} \bra{\psi'} \op{H}_0^{\lambda} \ket{\psi'} = \bra{\psi_{\lambda}} \op{H}_0^{\lambda} \ket{\psi_{\lambda}}.
         \end{equation}
\end{theorem}
This form of the theorem is relevant for the adiabatic-connection expression for the exchange--correlation energy~\cite{Langreth1975,Gunnarsson1976,Gunnarsson1977,Langreth1977}. The adiabatic-connection expression is essentially the integral
\begin{equation}
   \Fldft^{\lambda=1}(\rho,\mathbf{p},\mathbf{L}_{\mathbf{G}}) - \Fldft^{\lambda=0}(\rho,\mathbf{p},\mathbf{L}_{\mathbf{G}}) = \int_0^1 \frac{d \Fldft^{\lambda}(\rho,\mathbf{p},\mathbf{L}_{\mathbf{G}})}{d\lambda} d\lambda = \int_0^1 \bra{\psi_{\lambda}} W \ket{\psi_{\lambda}} d\lambda,
\end{equation}
where the last equality relies on the Hellmann--Feynman theorem and {\it the adiabaticity assumption} that $\Fldft^{\lambda}$ is differentiable with respect to $\lambda$---in 
fact, it is sufficient that the left- or right-derivative with respect to $\lambda$ exists. 
This form, combined with approximate Lieb optimization to determine the scalar potential needed to enforce the density constraint, 
has been used to characterize the exchange--correlation functional in standard DFT~\cite{COLONNA_JCP110_2828,FRYDEL_JCP112_5292,SAVIN_JCP115_6827,WU_JCP118_2498,TEALE_JCP130_104111,TEALE_JCP132_164115,TEALE_JCP133_164112,STROMSHEIM_JCP135_194109}.

Furthermore, we introduce the density-matrix functional
\begin{equation}
        \Fldftdm^{\lambda}(\rho,\mathbf{p},\mathbf{L}_{\mathbf{G}}) := \inf_{\Gamma \mapsto \rho,\mathbf{p},\mathbf{L}_{\mathbf{G}}} \trace{ \Gamma \op{H}_0^{\lambda}}.
\end{equation}
Although the above proof does not apply to this functional so that its expectation-valuedness is presently not established, it will be useful in later discussions---see Sec.~\ref{secKS}.

\subsection{Mixed-state $N$-representability}
\label{secLDFTNREP}

In this section, we focus on states with zero spin $\mathbf{S} = \mathbf{0}$, which is sufficient to establish a general $N$-representability result for LDFT. Mixed-state $N$-representability has previously been established very generally for pairs $(\rho,\jpvec)$~\cite{TELLGREN_PRA89_012515}. The $N$-representability problem for $(\rho,\mathbf{p},\mathbf{L}_{\mathbf{G}})$ is therefore solved if one can always construct a current density $\jpvec$ that is compatible with the given triple. This can be achieved, for example, 
by minimization of the current-correction to the von Weiz\"acker energy, $\int |\jpvec|^2/2\rho d\mathbf{r}$, subject the constraints $\mathbf{p} = \int \jpvec d\mathbf{r}$ and $\boldsymbol{\Lambda} = \int \mathbf{r}_{\mathbf{R}} \times \jpvec d\mathbf{r}$. The minimizer is a velocity field of the form
\begin{equation}
   \frac{1}{2} \boldsymbol{\kappa} = \frac{\jpvec}{\rho} = \boldsymbol{\zeta} + \frac{1}{2} \boldsymbol{\nu}  \times \mathbf{r}_{\mathbf{R}},
\end{equation}
where the vector field $\boldsymbol{\kappa}$ (in the notation in Ref.~\cite{TELLGREN_PRA89_012515}) is twice the paramagnetic velocity, and the constants (Lagrange multipliers) $\boldsymbol{\zeta}$ and $\boldsymbol{\nu}$ are chosen so as to satisfy the constraints. The latter multiplier is in fact the vorticity of this particular velocity field. The equations to determine the multipliers are straightforward,
\begin{equation}
  \mathbf{p} = \int \jpvec d\mathbf{r} = N \boldsymbol{\zeta},
\end{equation}
\begin{equation}
  \boldsymbol{\Lambda} = \int \mathbf{r}_{\mathbf{R}}\times \jpvec d\mathbf{r} = \frac{1}{2} \int \rho \mathbf{r}_{\mathbf{R}} \times (\boldsymbol{\nu} \times \mathbf{r}_{\mathbf{R}}) d\mathbf{r} = \frac{1}{2} \int \rho \, \big( |\mathbf{r}_{\mathbf{R}}|^2 \mathbf{I} - \mathbf{r}_{\mathbf{R}} \mathbf{r}_{\mathbf{R}}^{\mathrm{T}}) \boldsymbol{\nu} d\mathbf{r} = \frac{1}{2} \mathbf{Q} \boldsymbol{\nu},
\end{equation}
where the last equality defines the moment-of-inertia tensor, $\mathbf{Q}$.

Inserting our specific $\boldsymbol{\kappa}$ into the general density-matrix construction of Sec.~IV.A of Ref.~\cite{TELLGREN_PRA89_012515}, we obtain
\begin{align}
  P_{\lambda}(\mathbf{r},\mathbf{s}) & = \sqrt{\rho(\mathbf{r}) \, \rho(\mathbf{s})} e^{-\lambda |\mathbf{r}-\mathbf{s}|^2} e^{2i \boldsymbol{\zeta}\cdot(\mathbf{r}-\mathbf{s})}, 
     \\
  Q_{\mu}(\mathbf{r},\mathbf{s}) & = \sqrt{\rho(\mathbf{r}) \, \rho(\mathbf{s})} e^{-\mu |\mathbf{r}-\mathbf{s}|^2} e^{-i\boldsymbol{\nu}\cdot(\mathbf{r}_{\mathbf{R}}\times\mathbf{s}_{\mathbf{R}}) - |\boldsymbol{\nu}\times(\mathbf{r}-\mathbf{s})|^2/16\mu},
       \\
  D_{\lambda\mu}(\mathbf{r},\mathbf{s}) & = \frac{1}{2} P_{\lambda}(\mathbf{r},\mathbf{s}) + \frac{1}{2} Q_{\lambda}(\mathbf{r},\mathbf{s}),
\end{align}
where $\lambda, \mu \geq \tfrac{2 p}{\pi} (\tfrac{1}{4} N \| \rho\|_q )^{2p/3}$ and $p,q > 1$ are H\"older conjugates. Because $\rho \in L^1 \cap L^3$, we may for simplicity fix, for example, $p = q = 2$. Here, we have modified the original form slightly by replacing $\mathbf{r}$ and $\mathbf{s}$ by $\mathbf{r}_{\mathbf{R}}$ and $\mathbf{s}_{\mathbf{R}}$, respectively, where it makes a difference in the exponentials; note that $\mathbf{r}-\mathbf{s} = \mathbf{r}_{\mathbf{R}} - \mathbf{s}_{\mathbf{R}}$.

The constant shift of the coordinates does not affect the results in Ref.~\cite{TELLGREN_PRA89_012515}, to which we refer for proofs that $D_{\lambda\mu}$ is a valid one-particle reduced density matrix (e.g., it is shown that spatial occupations numbers in the interval $[0,2]$). The density matrix $D_{\lambda\mu}$ reproduces $(\rho,\jpvec)$ and therefore also $\mathbf{p}$, $\boldsymbol{\Lambda}$, and $\mathbf{L}_{\mathbf{G}}$. Given our choice of function spaces, this construction thus yields a strong $N$-representability result for LDFT:
\begin{theorem}
  \label{thmMIXEDNREP}
  If $\rho$ is mixed-state $N$-representable, then so is $(\rho,\mathbf{p},\mathbf{L}_{\mathbf{G}})$ for all $(\mathbf{p},\mathbf{L}_{\mathbf{G}}) \in \mathbb{R}^6$.
\end{theorem}
The canonical kinetic energy of this explicit construction is given by
\begin{equation}
 \begin{split}
  T & = T_{\text{W}} + \frac{|\mathbf{p}|^2}{N} + 2 \boldsymbol{\Lambda}^{\mathrm{T}} \mathbf{Q}^{-1} \boldsymbol{\Lambda} + \frac{1}{2 \mu} N |\mathbf{Q}^{-1} \boldsymbol{\Lambda}|^2 + \frac{1}{2} (\lambda + \mu) N,
 \end{split}
\end{equation}
where $T_{\text{W}} = \int |\nabla \rho|^2 / 8\rho d\mathbf{r}$ is the von Weizs\"acker kinetic energy, setting an upper limit on the density-matrix functional $\Fldftdm^{\lambda=0}(\rho,\mathbf{p},\mathbf{L}_{\mathbf{G}})$.

\subsection{Hohenberg--Kohn-like results}
\label{secLDFT_HK}

Paramagnetic CDFT admits a HK-like result in the form that $(v,\mathbf{A})$-representable ground-state densities $(\rho,\jmvec)$ determine the ground-state wave function, provided that it is non-degenerate~\cite{VIGNALE_PRL59_2360,CAPELLE_PRB65_113106}. In the degenerate case, we can consider two Hamiltonians $\bar{H}(u_k,\mathbf{A}_k)$, $k=1,2$, with ground-states $\Gamma_k$. If the ground states share the same densities $(\rho_1,\jmvecix{1}) = (\rho_2,\jmvecix{2})$, then $\Gamma_1$ is also a ground state of $\bar{H}(u_2,\mathbf{A}_2)$ and vice versa (see Sec.~III.B in Ref.~\onlinecite{TELLGREN_PRA86_062506}). The literature proofs are given for $\jpvec$ (or $\gspin=0$) but hold also for $\jmvec$ with trivial changes.  In fact, HK-like results are possible whenever potentials are paired linearly with densities in the ground-state energy expression. We term this type of results {\it weak HK-like results}, as they are weaker than the HK theorem in standard DFT, where also the potentials are determined by the ground-state densities.

Below, we state three weak HK-like results specialized to the LDFT setting, all equally valid with ($\gspin=2$) and without ($\gspin=0$) inclusion of the spin-Zeeman term in the Hamiltonian.
\begin{theorem}
  Let two or more potential triples $(u_k,\mathbf{a}_k,\mathbf{B}_k)$, with $k = 1, \ldots, K \geq 2$, be given. Let $\Gamma_k$ denote a pure or mixed ground state of $\bar{H}(u_k,\mathbf{a}_k,\mathbf{B}_k)$. Suppose that all $\Gamma_k$ give rise to the same density triples $(\rho_{\Gamma_k}, \mathbf{p}_{\Gamma_k}, \mathbf{L}_{\mathbf{G}; \Gamma_k}) = (\rho, \mathbf{p}, \mathbf{L}_{\mathbf{G}})$. Then $\Gamma_k$ is also a ground state of $\bar{H}(u_l,\mathbf{a}_l,\mathbf{B}_l)$, for all $k,l$.
\end{theorem}
\begin{proof}
The energy that the $l$th Hamiltonian assigns to the $k$th state can be written as
 \begin{equation}
   \trace{\Gamma_k \bar{H}(u_l,\mathbf{a}_l,\mathbf{B}_l) } = (u_l|\rho) + \mathbf{a}_l\cdot\mathbf{p} + \frac{1}{2} \mathbf{B}_l\cdot\mathbf{L}_{\mathbf{G}} + \trace{\Gamma_k \op{H}_0 }.
 \end{equation}
Adding the same energy expression with the indices interchanged, we obtain 
 \begin{equation}
   \trace{\Gamma_k \bar{H}(u_l,\mathbf{a}_l,\mathbf{B}_l) } + \trace{\Gamma_l \bar{H}(u_k,\mathbf{a}_k,\mathbf{B}_k) } = \trace{\Gamma_k \bar{H}(u_k,\mathbf{a}_k,\mathbf{B}_k) } + \trace{\Gamma_l \bar{H}(u_l,\mathbf{a}_l,\mathbf{B}_l) }.
 \end{equation}
Together with the variation principle $\trace{\Gamma_k \bar{H}(u_l,\mathbf{a}_l,\mathbf{B}_l) } \geq \trace{\Gamma_l \bar{H}(u_l,\mathbf{a}_l,\mathbf{B}_l) }$, the result now follows.
\end{proof}

The above result can be generalized slightly by taking convex combinations of density matrices and potentials. The set of potential triples $\{(u,\mathbf{a},\mathbf{B})\}$ that give rise to a given density triple is in fact a convex set. Moreover, each such potential triple is a subgradient of the functional $\bar{F}$ at $(\rho,\mathbf{p},\mathbf{L}_{\mathbf{G}})$.

Under the additional assumption that the ground state is non-degenerate, we obtain a HK-like result in the more familiar form that ground-state densities $(\rho, \mathbf{p}, \mathbf{L}_{\mathbf{G}})$ determine the ground-state wave function.
\begin{theorem}
  Let two or more potential triples $(u_k,\mathbf{a}_k,\mathbf{B}_k)$ with $k = 1, \ldots, K \geq 2$ be given. Suppose that each Hamiltonian $\bar{H}(u_k,\mathbf{a}_k,\mathbf{B}_k)$ has a non-degenerate ground state $\psi_k$ that gives rise to the same density triple $(\rho_{\psi_k}, \mathbf{p}_{\psi_k}, \mathbf{L}_{\mathbf{G}; \psi_k}) = (\rho, \mathbf{p}, \mathbf{L}_{\mathbf{G}})$. Then $\psi_k$ and $\psi_l$ are equal up to a global phase for all $k,l$.
\end{theorem}
A corollary is that non-degenerate ground-state densities $(\rho, \mathbf{p}, \mathbf{L}_{\mathbf{G}})$ also determine the paramagnetic current density $\jpvec$ as well as the magnetization current density $\jmvec$.
We now recall that the eigenstates of the Hamiltonian $H(v,\mathbf{A})$ have divergence-free physical current densities, $\nabla\cdot\mathbf{j} = \nabla\cdot(\jpvec+\rho\mathbf{A}) = 0$ and a vanishing physical momentum. If two or more potential triples $(u_k,\mathbf{a}_k,\mathbf{B}_k)$ give rise to the same non-degenerate ground-state density triple $(\rho, \mathbf{p}, \mathbf{L}_{\mathbf{G}})$, and consequently the same ground-state wave function $\psi$ and paramagnetic current density $\jpvec$, then 
\begin{equation}
  \nabla\cdot(\jpvec + \rho (\mathbf{a}_k + \tfrac{1}{2} \mathbf{B}_k\times\mathbf{r}_{\mathbf{G}})) = 0.
\end{equation}
Hence, for all $k,l$, we have
\begin{equation}
  ((\mathbf{a}_k-\mathbf{a}_l) + \tfrac{1}{2} (\mathbf{B}_k-\mathbf{B}_l)\times\mathbf{r}_{\mathbf{G}})\cdot\nabla\rho = 0.
\end{equation}
The first term can be eliminated by exploiting the fact that $\mathbf{p}_k = \mathbf{p}_l = \mathbf{p}$ and the relation
\begin{equation}
   \boldsymbol{\pi}_k = \mathbf{p}_k + N \mathbf{a}_k + \frac{1}{2} \mathbf{B}_k\times\boldsymbol{\mu}_{\mathbf{G}} = 0.
\end{equation}
Inserting $\mathbf{a}_k = -N^{-1} \mathbf{p}_k - \frac{1}{2} N^{-1} \mathbf{B}_k\times\boldsymbol{\mu}_{\mathbf{G}}$, we obtain
\begin{equation}
  \tfrac{1}{2} \big( (\mathbf{B}_k-\mathbf{B}_l)\times(\mathbf{r}-\mathbf{R}) \big) \cdot\nabla\rho = 0,
\end{equation}
where $\mathbf{R}$ is the center of charge. If $\mathbf{B}_k\neq\mathbf{B}_l$, it follows that the ground-state density $\rho$ must be cylindrically symmetric about an axis passing through $\mathbf{R}$ and directed along $\mathbf{B}_k-\mathbf{B}_l$.

Hence, for non-cylindrically symmetric densities $\rho$, LDFT does admit a {\it strong HK-like result} in the form that the external potentials are determined:
\begin{theorem}
  Let two or more potential triples $(u_k,\mathbf{a}_k,\mathbf{B}_k)$ with $k = 1, \ldots, K \geq 2$ be given. Suppose that each Hamiltonian $\bar{H}(u_k,\mathbf{a}_k,\mathbf{B}_k)$ has a non-degenerate ground state $\psi_k$ with the same density $(\rho_{\psi_k}, \mathbf{p}_{\psi_k}, \mathbf{L}_{\mathbf{G}; \psi_k}) = (\rho, \mathbf{p}, \mathbf{L}_{\mathbf{G}})$. Suppose further that $\rho$ is not cylindrically symmetric about any axis. Then $\psi_k$ and $\psi_l$ are equal up to a global phase, $u_k = u_l + \text{const}$ and $(\mathbf{a}_k,\mathbf{B}_k) = (\mathbf{a}_l,\mathbf{B}_l)$, for all $k,l$.
\end{theorem}

To establish the HK-like result that $(\rho,\mathbf{p},\mathbf{L}_{\mathbf{G}})$ determines the external potentials, we exploited the fact that this triple 
determines $\jpvec$ (and the ground state $\psi$) under a non-degeneracy assumption. Intuitively, we expect $(\rho,\boldsymbol{\Lambda})$ in a similar fashion 
to determine the gauge-shift invariant components of the potentials---that is, $\mathbf{B}$ and $u$ to within a constant, leaving $\mathbf{a}$ undetermined. To investigate this, we must deal with the fact that a gauge shift in the magnetic vector potential results in a non-trivial change of the phase of the ground-state wave function, 
thereby affecting the paramagnetic current density.

\begin{theorem}
  Let two potential triples $(u_k,\mathbf{a}_k,\mathbf{B}_k)$ with $k = 1, 2$ be given. Suppose that each Hamiltonian $\bar{H}(u_k,\mathbf{a}_k,\mathbf{B}_k)$ has a ground state $\psi_k$ that gives rise to the same density pair $(\rho_{\psi_k}, \boldsymbol{\Lambda}_{\psi_k}) = (\rho, \boldsymbol{\Lambda})$. Then, for a constant $\mathbf{c}$ (determined below), the state $\psi_1' = \psi_2 e^{-i\mathbf{c}\cdot\sum_j \mathbf{r}_j}$ is a ground state of  $\bar{H}(u_1,\mathbf{a}_1,\mathbf{B}_1)$ and $\psi_2' = \psi_1 e^{i\mathbf{c}\cdot \sum_j \mathbf{r}_j}$ is a ground state of $\bar{H}(u_2,\mathbf{a}_2,\mathbf{B}_2)$.
\end{theorem}
\begin{proof}
Inserting the identity $\mathbf{L}_{k;\mathbf{G}} = \boldsymbol{\Lambda} + (\mathbf{R}-\mathbf{G})\times\mathbf{p}_k$ into the orbital Zeeman term, we obtain
\begin{equation}
 \begin{split}
  E'_1 & := \bra{\psi'_1} \bar{H}(u_1,\mathbf{a}_1,\mathbf{B}_1) \ket{\psi'_1} = \bra{\psi_2} \bar{H}(u_1,\mathbf{a}_1-\mathbf{c},\mathbf{B}_1) \ket{\psi_2}
        \\
 & = (u_1|\rho) + (\mathbf{a}_1-\mathbf{c})\cdot\mathbf{p}_2 + \frac{1}{2}\mathbf{B}_1\cdot\mathbf{L}_{2;\mathbf{G}} + \bra{\psi_2} \op{H}_0 \ket{\psi_2}
      \\
    & = (u_1|\rho) + (\mathbf{a}_1 + \tfrac{1}{2} \mathbf{B}_1\times(\mathbf{R}-\mathbf{G}) - \mathbf{c})\cdot\mathbf{p}_2 + \frac{1}{2}\mathbf{B}_1\cdot\boldsymbol{\Lambda} + \bra{\psi_2} \op{H}_0 \ket{\psi_2}.
  \end{split}
\end{equation}
Noting that $\mathbf{A}_k(\mathbf{R}) = \mathbf{a}_k + \tfrac{1}{2} \mathbf{B}_k\times(\mathbf{R}-\mathbf{G})$ and adding and subtracting the terms needed to form $E_2 := \bra{\psi_2} \bar{H}(u_2, \mathbf{a}_2, \mathbf{B}_2)$, we next find
\begin{equation}
 \begin{split}
   E'_1  & = E_2 + (u_1-u_2|\rho) - \mathbf{a}_2\cdot\mathbf{p}_2 - \frac{1}{2} \mathbf{B}_2\cdot \mathbf{L}_{2;\mathbf{G}} + (\mathbf{A}_1(\mathbf{R}) - \mathbf{c})\cdot\mathbf{p}_2 + \frac{1}{2}\mathbf{B}_1\cdot\boldsymbol{\Lambda}
      \\
    & = E_2 + (u_1-u_2|\rho) + (\mathbf{A}_1(\mathbf{R})-\mathbf{A}_2(\mathbf{R})-\mathbf{c})\cdot\mathbf{p}_2 + \frac{1}{2}(\mathbf{B}_1-\mathbf{B}_2)\cdot\boldsymbol{\Lambda}.
 \end{split}
\end{equation}
As in the standard HK proof, we obtain a new equation by interchanging the indices 1 and 2, being careful to also replace occurrences of $\mathbf{c}$ by $-\mathbf{c}$, yielding 
\begin{equation}
 \begin{split}
   E'_2  & := \bra{\psi'_2} \bar{H}(u_2,\mathbf{a}_2,\mathbf{B}_2) \ket{\psi'_2} = \bra{\psi_1} \bar{H}(u_2,\mathbf{a}_2+\mathbf{c},\mathbf{B}_2) \ket{\psi_1}
      \\
    & = E_1 + (u_2-u_1|\rho) + (\mathbf{A}_2(\mathbf{R})-\mathbf{A}_1(\mathbf{R})+\mathbf{c})\cdot\mathbf{p}_1 + \frac{1}{2}(\mathbf{B}_2-\mathbf{B}_1)\cdot\boldsymbol{\Lambda},
 \end{split}
\end{equation}
where $E_1 := \bra{\psi_1} \bar{H}(u_1,\mathbf{a}_1,\mathbf{B}_1) \ket{\psi_1}$. Adding the last two equations and setting $\mathbf{c} = \mathbf{A}_1(\mathbf{R})-\mathbf{A}_2(\mathbf{R})$, we obtain
\begin{equation}
    E'_1 + E'_2 = E_1 + E_2.
\end{equation}
By the variation principle, we also have $E'_1 \geq E_1$ and $E'_2 \geq E_2$. The only possibility is $E'_1 = E_1$ and $E'_2 = E_2$. Thus, both $\psi_1$ and $\psi'_1$ must be ground states of $\bar{H}(u_1,\mathbf{a}_1,\mathbf{B}_1)$. Likewise, both $\psi_2$ and $\psi'_2$ must be ground states of $\bar{H}(u_2,\mathbf{a}_2,\mathbf{B}_2)$.
\end{proof}

Under the additional non-degeneracy assumption, it follows that two Hamiltonians $\bar{H}(u_k,\mathbf{a}_k,\mathbf{B}_k)$ with the same ground-state density pair $(\rho,\boldsymbol{\Lambda})$ have ground states related by $\psi_2 = \psi_1 e^{i\mathbf{c}\cdot \sum_j \mathbf{r}_j}$ with  $\mathbf{c} = \mathbf{A}_1(\mathbf{R})-\mathbf{A}_2(\mathbf{R})$. Consequently, the paramagnetic current densities are related as 
\begin{equation}
  \jpvecix{2} = \jpvecix{1} + \rho \mathbf{c}.
\end{equation}
The corresponding physical current densities must be divergence free,
\begin{align}
   \nabla\cdot(\jpvecix{1} + \mathbf{a}_1 + \tfrac{1}{2} \mathbf{B}_1 \times \mathbf{r}_{\mathbf{G}}) & = 0,
         \\
   \nabla\cdot(\jpvecix{2} + \mathbf{a}_2 + \tfrac{1}{2} \mathbf{B}_2 \times \mathbf{r}_{\mathbf{G}}) & = 0.
\end{align}
Subtracting the first of these equations from the second, we obtain
\begin{equation}
   (\mathbf{c} + \mathbf{a}_2 - \mathbf{a}_1 + \tfrac{1}{2} (\mathbf{B}_2-\mathbf{B}_1) \times \mathbf{r}_{\mathbf{G}})\cdot\nabla\rho = (\tfrac{1}{2} (\mathbf{B}_2-\mathbf{B}_1) \times \mathbf{r}_{\mathbf{R}})\cdot\nabla\rho  = 0.
\end{equation}
If $\mathbf{B}_1 \neq \mathbf{B}_2$, it follows that $\rho$ is cylindrically symmetric about an axis passing through the center of charge in the direction $\mathbf{B}_2-\mathbf{B}_1$. Hence, only components of the magnetic field parallel to a cylindrical symmetry axis are undetermined. To summarize, we have proved the following theorem:
\begin{theorem}
  Let two or more potential triples $(u_k,\mathbf{a}_k,\mathbf{B}_k)$ with $k = 1, \ldots, K \geq 2$ be given. Suppose that each Hamiltonian $\bar{H}(u_k,\mathbf{a}_k,\mathbf{B}_k)$ has a non-degenerate ground state $\psi_k$ with the same density triple $(\rho_{\psi_k}, \boldsymbol{\Lambda}_{\psi_k}) = (\rho, \boldsymbol{\Lambda})$. Suppose further that $\rho$ is not cylindrically symmetric about any axis. Then $\psi_k = \psi_l e^{i(\mathbf{A}_k(\mathbf{R})-\mathbf{A}_l(\mathbf{R}))\cdot\sum_j \mathbf{r}_j}$ (to within a global phase) and $\mathbf{B}_k = \mathbf{B}_l$, for all $k,l$.
\end{theorem}

A brief remark can be made about the case when there is a rotational symmetry axis. From symmetry considerations, we conclude that projection of the rotation axes from $\mathbf{B}_k$ and $\boldsymbol{\Lambda}$ results in (anti-)parallel vectors. Let $\mathcal{C}$ denote the projection onto the space spanned by all rotational symmetry axes and $\mathcal{C}_{\perp}$ the projection onto the complement. Then
\begin{equation}
   \mathcal{C}_{\perp} \boldsymbol{\Lambda} \parallel \mathcal{C}_{\perp} \mathbf{B}.
\end{equation}

Another triple that determines the external magnetic field is the gauge-shift invariant $(\rho,\boldsymbol{\Lambda}, \mathbf{J}_{\mathbf{R}})$, where $\mathbf{J}_{\mathbf{R}}$ is the \magm{physical magnetic momentum} with respect to the center of charge $\mathbf{R}$,
\begin{equation}
 \begin{split}
  \mathbf{J}_{\mathbf{R}} & = \mathbf{L}_{\mathbf{R}} + \int \rho(\mathbf{r}) \, (\mathbf{r}-\mathbf{R}) \times (\mathbf{a} + \tfrac{1}{2} \mathbf{B} \times (\mathbf{r} - \mathbf{G})) d\mathbf{r}  \\
    & = \mathbf{L}_{\mathbf{R}} + \int \rho(\mathbf{r}) \, (\mathbf{r}-\mathbf{R}) \times (\tfrac{1}{2} \mathbf{B} \times (\mathbf{r} - \mathbf{R})) d\mathbf{r} + \int \rho(\mathbf{r}) \, (\mathbf{r}-\mathbf{R}) \times (\mathbf{a} + \tfrac{1}{2} \mathbf{B} \times (\mathbf{R} - \mathbf{G})) d\mathbf{r}.
 \end{split}
\end{equation}
The last term vanishes because $\boldsymbol{\mu}_{\mathbf{R}} = \int (\mathbf{r}-\mathbf{R}) \rho d\mathbf{r} = 0$ is the dipole moment with respect to the center of charge. Rewriting the other integral and noting that $\mathbf{L}_{\mathbf{R}} = \boldsymbol{\Lambda}$, we obtain
\begin{equation}
 \begin{split}
   \mathbf{J}_{\mathbf{R}} & = \boldsymbol{\Lambda} + \frac{1}{2} \int \rho \big( |\mathbf{r}_{\mathbf{R}}|^2 \mathbf{B} - (\mathbf{B}\cdot\mathbf{r}_{\mathbf{R}}) \, \mathbf{r}_{\mathbf{R}} \big) d\mathbf{r}.
 \end{split}
\end{equation}
Recalling the definition of the moment-of-inertia tensor, $\mathbf{Q} = \int (|\mathbf{r}_{\mathbf{R}}|^2 \mathbf{I} - \mathbf{r}_{\mathbf{R}} \mathbf{r}_{\mathbf{R}}^{\mathrm{T}}) \rho d\mathbf{r}$, it follows that
\begin{equation}
    \mathbf{J}_{\mathbf{R}} - \boldsymbol{\Lambda} = \frac{1}{2} \mathbf{Q} \mathbf{B}.
\end{equation}
The inertia tensor $\mathbf{Q}$ is symmetric and it is therefore always possible to rotate the coordinate axes so that $\mathbf{Q}$ becomes diagonal. Without loss of generality, we therefore write
\begin{equation}
   \mathbf{Q} = \int \rho \begin{pmatrix} y_{\mathbf{R}}^2 + z_{\mathbf{R}}^2 & 0 & 0 \\ 0 & x_{\mathbf{R}}^2 + z_{\mathbf{R}}^2 & 0 \\ 0 & 0 & x_{\mathbf{R}}^2 + y_{\mathbf{R}}^2 \end{pmatrix} d\mathbf{r}.
\end{equation}
This matrix is always positive definite and invertible. Hence, we find
\begin{equation}
  \label{eqBfromLJ}
   \mathbf{B} = 2\mathbf{Q}^{-1} (\mathbf{J}_{\mathbf{R}} - \boldsymbol{\Lambda}).
\end{equation}
To determine the magnetic field, it was not even necessary to assume that $(\rho,\boldsymbol{\Lambda}, \mathbf{J}_{\mathbf{R}})$ comes from a ground state. As long as $(\rho,\boldsymbol{\Lambda})$ has been obtained from some state and $\mathbf{J}_{\mathbf{R}}$ has been calculated from the same state and some additional vector potential $\mathbf{A} = \mathbf{a} + \tfrac{1}{2} \mathbf{B} \times (\mathbf{r}-\mathbf{G})$, the above expression for $\mathbf{B}$ is valid.

\subsection{Constructive relationship between $\mathbf{B}$ and ground-state densities $(\rho,\jmvec)$}

The above HK-like results have been established using proofs by contradiction. 
Going beyond the strict LDFT framework and studying the current density, we may derive a more explicit, constructive relationship. In particular, for the density pair
$(\rho,\jmvec)$, arising from a ground state of some potential $(u,\mathbf{A}=\mathbf{a}+\frac{1}{2} \mathbf{B}\times\mathbf{r}_{\mathbf{G}})$, a constructive relationship can be established. 

First note that the condition of a vanishing physical momentum can be expressed as $\boldsymbol{\pi} = \mathbf{p} + N \mathbf{A}(\mathbf{R}) + \frac{1}{2} \mathbf{B}\times\boldsymbol{\mu}_{\mathbf{R}}$. Because the dipole moment relative to the center of charge vanishes, we find that the canonical momentum determines the vector potential at the center of charge, $\mathbf{p} = -\mathbf{A}(\mathbf{R})/N$. The condition that the physical current is divergence-free can then be expressed as
\begin{equation}
  \frac{1}{2} (\mathbf{r}_{\mathbf{R}}\times\nabla\rho)^{\mathrm{T}} \mathbf{B} = -\nabla\cdot(\jmvec + \rho \mathbf{A}(\mathbf{R})) = -\nabla\cdot(\jmvec - \rho \mathbf{p}/N),
\end{equation}
where $(\mathbf{r}_{\mathbf{R}}\times\nabla\rho)^{\mathrm{T}}$ is a row vector. Multiplying by the column vector $\rho^{-1} \mathbf{r}_{\mathbf{R}}\times\nabla\rho$ and integrating over all space, we obtain 
\begin{equation}
   \frac{1}{2} \mathcal{M} \mathbf{B} = \frac{1}{2} \Big( \int (\mathbf{r}_{\mathbf{R}}\times\nabla\rho) \frac{1}{\rho} (\mathbf{r}_{\mathbf{R}}\times\nabla\rho)^{\mathrm{T}} d\mathbf{r} \Big) \mathbf{B} = -\int \frac{1}{\rho} (\mathbf{r}_{\mathbf{R}}\times\nabla\rho) \nabla\cdot(\jmvec + \rho \mathbf{A}(\mathbf{R})) d\mathbf{r},
\end{equation}
where the first equality defines $\mathcal{M}$. The rationale for the factor $1/\rho$ is simply that it makes each element of $\mathcal{M}$ bounded by a weighted von Weiz\"acker energy. Therefore, $\mathcal{M}$ is guaranteed to be finite, given our choice of function spaces. As long as the resulting integral is well defined, the factor $1/\rho$ can be replaced by any function of $\rho$. An alternative expression for $\mathcal{M}$, which is somewhat more conducive to analysis, is the form
\begin{equation}
   \mathcal{M}  = 4 \int (\mathbf{r}_{\mathbf{R}}\times\nabla \sqrt{\rho}) (\mathbf{r}_{\mathbf{R}}\times\nabla \sqrt{\rho})^{\mathrm{T}} d\mathbf{r}.
\end{equation}
It can be verified (by writing out the expression in cylindrical coordinates) that $\mathcal{M} = 0$ for a spherically symmetric density with $\mathrm{rank}(\mathcal{M})=2$ when there is a single cylindrical symmetry axis. For a non-cylindrically symmetric density, $\mathrm{rank}(\mathcal{M})=3$. Hence, for non-cylindrically symmetric densities $\rho$,
\begin{equation}
   \mathbf{B} = -2 \mathcal{M}^{-1} \int \frac{1}{\rho} (\mathbf{r}_{\mathbf{R}}\times\nabla\rho) \nabla\cdot(\jmvec + \rho \mathbf{A}(\mathbf{R})) d\mathbf{r} = -2 \mathcal{M}^{-1} \int \frac{1}{\rho} (\mathbf{r}_{\mathbf{R}}\times\nabla\rho) \nabla\cdot\widetilde{\mathbf{j}}_{\mathrm{m}} d\mathbf{r},
\end{equation}
where $\widetilde{\mathbf{j}}_{\mathrm{m}} = \jmvec + \rho \mathbf{A}(\mathbf{R}) = \jmvec - \rho \frac{\mathbf{p}}{N}$ is the deviation of $\jmvec$ from its spatially averaged value $\mathbf{p} = \int \jmvec d\mathbf{r}$.

An alternative constructive relationship is obtained by multiplying
the condition $\nabla\cdot\mathbf{j}$ by a coordinate product
$(r_a-R_a) (r_b-R_b)$, with $a,b \in \{x,y,z\}$, and integrating over all space:
\begin{equation}
  0 = \int r_{\mathbf{R}a} r_{\mathbf{R}b} \nabla_c j_c d\mathbf{r} = -\int (r_{\mathbf{R}b} \delta_{ac} + r_{\mathbf{R}a} \delta_{bc}) (\jmcomp{c} + \rho A_c(\mathbf{R}) + \frac{1}{2} \rho \epsilon_{cde} B_d r_{\mathbf{R}e}) d\mathbf{r}.
\end{equation}
Noting that the term proportional to $A_c(\mathbf{R})$ vanishes by virtue of being a constant times components of the dipole moment relative to the center of charge, and moving the $B_d$-dependent terms to the left-hand side, we obtain
\begin{equation}
  \label{eqLPOSandB}
  \frac{1}{2} \int \rho (r_{\mathbf{R}b} \epsilon_{ade} r_{\mathbf{R}e} + r_{\mathbf{R}a} \epsilon_{bde} r_{\mathbf{R}e}) B_d d\mathbf{r} =  -\int (r_{\mathbf{R}b} \jmcomp{a} + r_{\mathbf{R}a} \jmcomp{b}) d\mathbf{r}.
\end{equation}
When $a \neq b$, the right-hand side resembles the two terms in the \magm{magnetic-moment} component $\Lambda_c$ (with $c \neq a$ and $c \neq b$) but with different
signs: this expression is symmetric in $a,b$, whereas the \magm{magnetic moment} is anti-symmetric. Contraction by the Levi-Civita tensor $\epsilon_{cab}$ yields zero on both sides. We therefore contract by the absolute value of each element, $|\epsilon_{cab}|$, to obtain
\begin{equation}
  \label{eqLPOSandB2}
  \int \rho |\epsilon_{cab}| \, \epsilon_{ade} r_{\mathbf{R}b} r_{\mathbf{R}e} B_d d\mathbf{r} =  -2 \int |\epsilon_{cab}| r_{\mathbf{R}a} \jmcomp{b} d\mathbf{r}.
\end{equation}
We write the right-hand side (excluding sign) more compactly using the notation,
\begin{equation}
  L^{+}_{\mathbf{R};c} = \int (r_{\mathbf{R}b} \jmcomp{a} + r_{\mathbf{R}a} \jmcomp{b}) d\mathbf{r} = |\epsilon_{cab}| \int r_{\mathbf{R}a} \jmcomp{b} \, d\mathbf{r},
\end{equation}
where $|\epsilon_{cab}|$ means the absolute value of the Levi-Civita tensor.
The left-hand sides of Eqs.~\eqref{eqLPOSandB} and \eqref{eqLPOSandB2} define a linear transformation of the magnetic field. This transformation can therefore be represented by a matrix,
\begin{equation}
  Q^{+}_{cd} = \int \rho (r_{\mathbf{R}b} \epsilon_{ade} r_{\mathbf{R}e} + r_{\mathbf{R}a} \epsilon_{bde} r_{\mathbf{R}e}) d\mathbf{r} = \int \rho |\epsilon_{cab}| r_{\mathbf{R}b} \epsilon_{ade} r_{\mathbf{R}e} \, d\mathbf{r},
\end{equation}
or, more explicitly,
\begin{equation}
  \mathbf{Q}^+ = \int \rho \, 
    \begin{pmatrix}
       y_{\mathbf{R}}^2-z_{\mathbf{R}}^2 & -x_{\mathbf{R}} y_{\mathbf{R}} & x_{\mathbf{R}} z_{\mathbf{R}} \\
       x_{\mathbf{R}} y_{\mathbf{R}} & z_{\mathbf{R}}^2-x_{\mathbf{R}}^2 & -y_{\mathbf{R}} z_{\mathbf{R}} \\
       -x_{\mathbf{R}} z_{\mathbf{R}} & y_{\mathbf{R}} z_{\mathbf{R}} & x_{\mathbf{R}}^2-y_{\mathbf{R}}^2
    \end{pmatrix} d\mathbf{r}.
\end{equation}
The same terms appear in $\mathbf{Q}^+$ and in the moment-of-inertia tensor $\mathbf{Q}$; however, half the terms have different signs. 
Alignment of the coordinate to, for example, the principal axes makes the off-diagonal elements of $\mathbf{Q}^+$ vanish. We then see that this matrix has full rank if, in such a coordinate system, $\int \rho x_{\mathbf{R}}^2 d\mathbf{r}$, $\int \rho y_{\mathbf{R}}^2 d\mathbf{r}$, and $\int \rho z_{\mathbf{R}}^2 d\mathbf{r}$ are all distinct. This fails if the system has cylindrical symmetry. Unlike $\mathcal{M}$ above, $\mathbf{Q}^+$ can also be rank-deficient in other cases, for example, when there is a 4-fold rotation symmetry axis rather than a continuous rotational symmetry. When the matrix has full rank, then
\begin{equation}
  \mathbf{B} = -2 (\mathbf{Q}^+)^{-1} \mathbf{L}_{\mathbf{R}}^+.
\end{equation}

\section{Kohn--Sham theory for LDFT}
\label{secKS}

Kohn--Sham theory is formulated in a similar manner for DFT, CDFT, and LDFT, the fundamental distinction being what the Kohn--Sham system is required to reproduce. In CDFT, the Kohn--Sham system is a Slater determinant that reproduces $(\rho,\jmvec)$. In LDFT, it is a Slater determinant which reproduces $(\rho,\mathbf{p},\mathbf{L}_{\mathbf{G}})$. More formally, we introduce the functionals
\begin{align}\label{eq:TsCSDT}
   T^{\mathrm{CDFT}}_{\mathrm{s}}(\rho,\jmvec) & = \inf_{\Phi \mapsto \rho,\jmvec} \bra{\Phi} \op{T} \ket{\Phi}, 
           \\
   T^{\mathrm{LDFT}}_{\mathrm{s}}(\rho,\mathbf{p},\mathbf{L}_{\mathbf{G}}) & = \inf_{\Phi \mapsto \rho,\mathbf{p},\mathbf{L}_{\mathbf{G}}} \bra{\Phi} \op{T} \ket{\Phi},
   \label{eq:TsLDFT}
\end{align}
where the minimization is over all Slater determinants $\Phi \in \WMSobolev{\mathbf{0}}$. The minimizers define the Kohn--Sham systems. A standard decomposition of the universal energy now yields
\begin{align}
   \Fvrdm(\rho,\jmvec) & = T_{\mathrm{s}}(\rho,\jmvec) + J(\rho) + F_{\text{xc}}(\rho,\jmvec),
           \\
   \Fldftdm(\rho,\mathbf{p},\mathbf{L}_{\mathbf{G}}) & = T_{\mathrm{s}}(\rho,\mathbf{p},\mathbf{L}_{\mathbf{G}}) + J(\rho) + F_{\text{xc}}(\rho,\mathbf{p},\mathbf{L}_{\mathbf{G}}),
    \label{eqKSDECOMPPRELIM}
\end{align}
where we drop the superscripts `CDFT' and `LDFT' and $J(\rho)$ is the Hartree energy---that is, the classical self-repulsion of the charge distribution.

A practical point related to the different constraints on the Kohn--Sham system is the calculation of magnetic properties. All properties that can be obtained as simple integrals over the current density are easy to compute in the CDFT Kohn--Sham framework. In the LDFT Kohn--Sham framework, however, only the linear momenta and \magm{magnetic moments} are reproduced---the Kohn--Sham current density may differ from the current density of the interacting system. As a result, other magnetic properties are less accessible in LDFT than in CDFT.

\subsection{Non-interacting $N$- and $(v,\mathbf{A})$-representability}

The non-interacting $N$-representability problem is one of cases where
the orbital and spin contributions to the \magm{magnetic moment} need to
be treated differently. Below, we focus on the case $\gspin=0$ with
only orbital contributions.

The Kohn--Sham theory of paramagnetic CDFT is known to suffer from
$N$-representability problems in the case for $N=1$. More
precisely, a single Kohn--Sham orbital is only capable of reproducing
densities $(\rho,\jpvec)$ with vanishing paramagnetic vorticity
$\nabla\times(\jpvec/\rho) = \mathbf{0}$ except possibly at isolated
points where the curl is singular. The situation for $N=2$ and $N=3$
is not fully settled and for $N\geq 4$ Lieb and Schrader have
shown that all density pairs satisfying some regularity conditions are
representable by a single Slater determinant~\cite{LIEB_PRA88_032516}.

When $N=1$, LDFT shares some of the $N$-representability problems of paramagnetic CDFT. This follows because a one-particle
wave function only has a non-vanishing canonical angular momentum when
there are nodal lines~\cite{RIESS_AP57_301,HIRSCHFELDER_JCP67_5477}. A
simple corollary is that $(\rho,\boldsymbol{\Lambda})$ is not
representable by a single orbital when $\rho(\mathbf{r}) > 0$ and
$\boldsymbol{\Lambda} \neq 0$. Moreover, for all wave functions with a
cylindrically symmetric density
$\rho(\mathbf{r}) = |\psi(\mathbf{r})|^2$, the component of the
intrinsic momentum $\boldsymbol{\Lambda}$ parallel to the rotation
axis is quantized~(see Theorem~2 in
Ref.~\onlinecite{HIRSCHFELDER_JCP67_5477}). The case $N=2$ is
much harder to analyse and it is an open question if
there are any restriction on pairs $(\rho,\boldsymbol{\Lambda})$ that
can be represented by $N=2$ and $N=3$ orbitals. For $N\geq 4$, Lieb and
Schrader's Slater-determinant representability proof for
$(\rho,\jpvec)$ entails that all $(\rho,\boldsymbol{\Lambda})$ are
representable too. The representability problems in CDFT and LDFT are avoided 
if the Kohn--Sham orbitals are allowed to
be fractionally occupied (see Ref.~\onlinecite{TELLGREN_PRA89_012515} for the CDFT case and Theorem~\ref{thmMIXEDNREP} above for the LDFT case).

The issue of $(v,\mathbf{A})$-representability is less well characterized. However, two-dimensional counterexamples to $(v,\mathbf{A})$-representability in CDFT have been constructed and these apply directly also to LDFT~\cite{TAUT_PRA80_022517}.

We remark that $T_{\mathrm{s}}(\rho,\mathbf{p},\mathbf{L}_{\mathbf{G}})$ and the functional $\Fldft^{\lambda=0}(\rho,\mathbf{p},\mathbf{L}_{\mathbf{G}})$ introduced in Theorem~\ref{thmGenExpVal} coincide on the set of non-degenerate, non-interacting $(v,\mathbf{A})$-representable densities. This fact follows since the relevant Hamiltonian is in 
this case a pure one-electron operator, implying that the non-degenerate ground state is a Slater determinant. 
In general, however, they are related through $T_{\mathrm{s}}(\rho,\mathbf{p},\mathbf{L}_{\mathbf{G}}) \geq \Fldft^{\lambda=0}(\rho,\mathbf{p},\mathbf{L}_{\mathbf{G}})$. 
The functional $\Fldftdm^{\lambda=0}$ is the kinetic energy in an ensemble Kohn--Sham setting with fractional occupation numbers allowed. 
Clearly, $T_{\mathrm{s}}(\rho,\mathbf{p},\mathbf{L}_{\mathbf{G}}) \geq \Fldft^{\lambda=0}(\rho,\mathbf{p},\mathbf{L}_{\mathbf{G}})\geq 
\Fldftdm^{\lambda=0}(\rho,\mathbf{p},\mathbf{L}_{\mathbf{G}})$ must hold, and with equalities on the set of non-degenerate, non-interacting $(v,\mathbf{A})$-representable densities.

\subsection{Expectation-valuedness of the Kohn--Sham kinetic energy}

 From Theorem~6 in Ref.~\onlinecite{LAESTADIUS_JMC52_2581},
the infimum in the definition of the CDFT functional $T_{\mathrm{s}}$ in Eq.~\eqref{eq:TsCSDT} is attained by some determinant. This result is premised on noninteracting $N$-representability, which presently has been proven generally in the case of four or more Kohn--Sham orbitals~\cite{LIEB_PRA88_032516}. The same holds for the LDFT version of the functional defined in Eq.~\eqref{eq:TsLDFT}:
\begin{theorem}
   If $N\geq 4$, then there exists a Slater determinant $\Phi \in \WMSobolev{\mathbf{0}}$ such that $\Phi \mapsto \rho,\mathbf{p},\mathbf{L}_{\mathbf{G}}$ and
   \begin{equation}
   	T_s(\rho,\mathbf{p},\mathbf{L}_{\mathbf{G}}) = \langle \Phi| \op{T}| \Phi\rangle.
   \end{equation}
\end{theorem}
\begin{proof}
Theorem~\ref{thmGenExpVal} above with $\lambda=0$ establishes a minimizer $\Phi \in \WMSobolev{\mathbf{0}}$ with all specified properties, save for being a determinant. However, we can use the same argument as in the proof of Theorem~6 in Ref.~\onlinecite{LAESTADIUS_JMC52_2581} to conclude that $\Phi$ is a determinant of $N$ orthonormal orbitals.
\end{proof}

\subsection{The exchange--correlation functional, gauge-shift invariance, and additive separability}

In Vignale and Rasolt's paramagnetic CDFT, there is a universal functional $\Fvr(\rho,\jmvec)$, which depends on the gauge-dependent paramagnetic current density. This gauge-dependence is necessary in their formalism since their universal functional includes the canonical kinetic energy. However, the exchange--correlation contribution to $\Fvr(\rho,\jmvec)$ must be gauge invariant and should therefore only depend on the gauge-invariant ``components'' of $(\rho,\jmvec)$. These components are the density itself and the paramagnetic vorticity $\boldsymbol{\nu} := \nabla\times(\jmvec/\rho)$. In LDFT, a minimalistic version of the same problem appears: The universal functional depends on the full triple $(\rho,\mathbf{p},\mathbf{L}_{\mathbf{G}})$, but the exchange--correlation contribution should only depend on the ``components'' that are invariant under the restricted gauge transformation that are allowed---namely, constant shifts of the vector potential. By counting the degrees of freedom 
(6 in $(\mathbf{p},\mathbf{L}_{\mathbf{G}})$ and 3 in $\mathbf{a}$), we conclude that these gauge-shift invariant components are contained in the pair $(\rho,\boldsymbol{\Lambda})$. The fact that this pair is sufficient to uniquely determine $\mathbf{B}$ for a non-cylindrically symmetric ground-state density supports the conclusion that there are no additional gauge-shift invariant ``components'' of $(\rho,\mathbf{p},\mathbf{L}_{\mathbf{G}})$.

From these considerations, we may thus decompose the universal functional as
\begin{equation}
  \Fldftdm(\rho,\mathbf{p},\mathbf{L}_{\mathbf{G}}) = T_{\mathrm{s}}(\rho,\mathbf{p},\mathbf{L}_{\mathbf{G}}) +  J(\rho) + F_{\text{xc}}(\rho,\boldsymbol{\Lambda}),
\end{equation}
where all the gauge-shift dependence is contained in the canonical kinetic energy term, $T_{\mathrm{s}}(\rho,\mathbf{p},\mathbf{L}_{\mathbf{G}})$, 
the arguments to $F_{\text{xc}}$ having been reduced to a gauge-shift independent density pair (compare Eq.~\eqref{eqKSDECOMPPRELIM}).

Another type of constraint comes from requirement of {\it additive
  separability}---that is, the requirement that the energy is additive
over well-separated non-interacting
subsystems~\cite{HELGAKER00}. Additive separability is related to the
stronger notions of size extensivity and size
consistency~\cite{BARTLETT_CPL50_190,POPLE_IJQC10_1,BARTLETT_IJQC14_561}. This
type of property is a desirable constraint on any practical
approximate exchange--correlation functional. Consider a system
consisting of $S$ subsystems. Each subsystem is a ground state of some
Hamiltonian $H(v_{\alpha},\mathbf{a},\mathbf{B})$ with density
$(\rho_{\alpha}, \mathbf{p}_{\alpha},
\mathbf{L}_{\alpha;\mathbf{G}})$ (denoting subsystems by Greek indices).
The total system is assumed to be the ground state of
$H(v, \mathbf{a}, \mathbf{B})$ with $v = \sum_{\alpha}
v_{\alpha}$. Then additive separability amounts to the condition
\begin{equation}
  F_{\text{xc}}(\rho,\boldsymbol{\Lambda}) = \sum_{\alpha=1}^S F_{\text{xc}}(\rho_{\alpha},\boldsymbol{\Lambda}_{\alpha}).
\end{equation}

Let $N_{\alpha}$ denote the number of electrons and $\mathbf{R}_{\alpha}$ the center of charge of subsystem $\alpha$. We shall assume that the subsystems are well 
separated---that is, that the distances $|\mathbf{R}_{\alpha} - \mathbf{R}_{\beta}|$, with $\alpha \neq \beta$, are arbitrarily large 
compared with the extent of each subsystem density, so that we can write
\begin{equation}
   \rho(\mathbf{r}) = \max_{1\leq \alpha \leq S} \rho_{\alpha}(\mathbf{r}).
\end{equation}
In other words, at any point $\mathbf{r}$, only a single subsystem contributes to the total density.

There is also a relation between the \magm{total intrinsic magnetic moment} and those of the individual subsystems. To derive a detailed form of this relation, we begin by introducing notation for moments of inertia with respect to multiple reference points,
\begin{equation}
  \mathbf{Q}_{\alpha,\mathbf{C}} = \int \rho_{\alpha}(\mathbf{r}) \, \big(|\mathbf{r}_{\mathbf{C}}|^2 \mathbf{I} - \mathbf{r}_{\mathbf{C}} \, \mathbf{r}_{\mathbf{C}}^{\mathrm{T}} \big) d\mathbf{r}.
\end{equation}
The moments of inertia of the total system are defined analogously. Writing simply $\mathbf{Q} = \mathbf{Q}_{\mathbf{R}}$ and $\mathbf{Q}_{\alpha} = \mathbf{Q}_{\alpha,\mathbf{R}_{\alpha}}$ for moments of inertia with respect to the relevant center of charge, we have the
relation
\begin{equation}
  \mathbf{Q} = \sum_{\alpha=1}^S \mathbf{Q}_{\alpha} + \sum_{\alpha=1}^S N_{\alpha} ( |\mathbf{d}_{\alpha}|^2 \mathbf{I} - \mathbf{d}_{\alpha} \mathbf{d}_{\alpha}^{\mathrm{T}}), \quad \mathbf{d}_{\alpha} = \mathbf{R}_{\alpha} - \mathbf{R}.
\end{equation}
Denoting the last term by $\mathcal{G}$, we have $\mathbf{Q} = \sum_{\alpha} \mathbf{Q}_{\alpha} + \mathcal{G}$. In general, $\mathrm{rank}(\mathcal{G}) = 2$ when $S=1$ or when the centres of charge of all subsystems are located on a line, and $\mathrm{rank}(\mathcal{G}) = 3$ otherwise. When $\mathcal{G}$ is invertible,
\begin{equation}
   \mathbf{Q}^{-1} \approx \mathcal{G}^{-1},
\end{equation}
with equality in the limit of infinite subsystem separations.

Next, we turn to the \magm{physical magnetic moment}, which by the ground-state assumption is independent of the reference point. The \magm{total physical magnetic moment} can therefore be written in the form
\begin{equation}
  \mathbf{J}_{\mathbf{R}} = \boldsymbol{\Lambda} + \frac{1}{2} \mathbf{Q} \mathbf{B}.
\end{equation}
\magm{The physical magnetic moment} is an extensive quantity and can thus also be written as a sum over subsystems,
\begin{equation}
  \mathbf{J}_{\mathbf{R}} = \sum_{\alpha=1}^S \mathbf{J}_{\alpha,\mathbf{R}} = \sum_{\alpha=1}^S \mathbf{J}_{\alpha,\mathbf{R}_{\alpha}},
\end{equation}
relying on the reference-point independence to move the reference point from $\mathbf{R}$ to $\mathbf{R}_{\alpha}$. Inserting $\mathbf{J}_{\alpha,\mathbf{R}_{\alpha}} = \boldsymbol{\Lambda}_{\alpha} + \frac{1}{2} \mathbf{Q}_{\alpha} \mathbf{B}$, we find
\begin{equation}
  \boldsymbol{\Lambda} = \sum_{\alpha=1}^S \boldsymbol{\Lambda}_{\alpha} + \frac{1}{2} \sum_{\alpha=1}^S \mathbf{Q}_{\alpha} \mathbf{B} - \frac{1}{2} \mathbf{Q} \mathbf{B} = \sum_{\alpha=1}^S \boldsymbol{\Lambda}_{\alpha} - \frac{1}{2} \mathcal{G} \mathbf{B}.
\end{equation}
When $\mathcal{G}$ has full rank, we obtain
\begin{equation}
  \label{eqSEPSUBSYSHK}
  \mathbf{Q}^{-1} \boldsymbol{\Lambda} \approx \mathcal{G}^{-1} \boldsymbol{\Lambda} \approx -\frac{1}{2} \mathbf{B},
\end{equation}
with equality for infinite subsystem separation. This is another HK-like result: {\it The \magm{total intrinsic magnetic moment} and momentum-of-inertia tensor 
of a large total system, with three or more subsystems not located on a single line, determine the external magnetic field.}

This result can also be understood by introducing an effective angular velocity $\boldsymbol{\omega} = \mathbf{Q}^{-1} \mathbf{J}_{\mathbf{R}}$. Recall that the angular velocity vector and the angular momentum of a classical rigid body are related through $\mathbf{J} = \mathbf{Q} \boldsymbol{\omega}$. Every point in a rigid body has the same angular velocity. When applied to a non-rigid body, $\boldsymbol{\omega}$ can be interpreted as an effective, distance-weighted angular velocity,
\begin{equation}
   \boldsymbol{\beta} = -2\mathbf{Q}^{-1} \boldsymbol{\Lambda} = \mathbf{B} - 2 \boldsymbol{\omega}.
\end{equation}
It is harder to sustain a significant effective angular velocity in a large system since it requires proportionally strong currents encircling the system. Consider, for example, a ring-shaped wire of radius $R$, wire cross section area $A$, constant electron density $\rho$ within the wire, and current $I$. The angular velocity is then given by $\omega = I/ (\rho A R)$, which vanishes in the limit of infinite $R$ and constant $\rho,A$, and $I$.

When all subsystems are located on a single line, the matrix $\mathcal{G}$ is of rank 2, with zero eigenvalue in the direction vector of the line. Hence, only the components of $\mathbf{B}$ in the plane perpendicular to the line are uniquely determined:
\begin{equation}
  \label{eqSEPSUBSYSHK2}
  \mathcal{G}^{-1} \boldsymbol{\Lambda} \approx -\frac{1}{2} \mathbf{B}_{\perp},
\end{equation}
where $\mathcal{G}^{-1}$ is now the generalized inverse and equality holds in the limit of infinite subsystem separations. To re-express the above relation in terms of $\mathbf{Q}$ instead of $\mathcal{G}$, we begin by writing $\mathcal{G} = \mathbf{U} \gamma \mathbf{I}_2 \mathbf{U}^{\mathrm{T}}$, where $\gamma = \sum_{\alpha} N_{\alpha} d_{\alpha}^2$ is the two-fold degenerate non-zero eigenvalue and $\mathbf{U}$ is a $3\times 2$ matrix containing the corresponding eigenvectors. The remaining eigenvector, with zero eigenvalue, is denoted $\hat{\mathbf{d}}$ and is by stipulation parallel to all $\mathbf{d}_{\alpha}$. Using the shorthand notation $\mathbf{S} = \sum_{\alpha} \mathbf{Q}_{\alpha}$ and the Woodbury matrix identity, we obtain
\begin{equation}
   \mathbf{Q}^{-1} = (\mathbf{S} + \gamma \mathbf{U} \mathbf{U}^{\mathrm{T}})^{-1} = \mathbf{S}^{-1} - \mathbf{S}^{-1} \mathbf{U} (\gamma^{-1} \mathbf{I}_2 + \mathbf{U}^{\mathrm{T}} \mathbf{S}^{-1} \mathbf{U})^{-1} \mathbf{U}^{\mathrm{T}} \mathbf{S}^{-1}.
\end{equation}
Some further algebra yields
\begin{equation}
 \begin{split}
   \mathbf{Q}^{-1} \mathcal{G}
   & = \mathbf{S}^{-1} \mathbf{U} (\gamma^{-1} \mathbf{I}_2 + \mathbf{U}^{\mathrm{T}} \mathbf{S}^{-1} \mathbf{U})^{-1}  \mathbf{U}^{\mathrm{T}}
  \approx \mathbf{S}^{-1} \mathbf{U} (\mathbf{U}^{\mathrm{T}} \mathbf{S}^{-1} \mathbf{U})^{-1}  \mathbf{U}^{\mathrm{T}},
 \end{split}
\end{equation}
with equality in the limit of infinite subsystem separations. Eq.~\eqref{eqSEPSUBSYSHK2} can now be expressed as
\begin{equation}
  \mathbf{U} \mathbf{U}^{\mathrm{T}} \mathbf{Q}^{-1} \boldsymbol{\Lambda} \approx -\frac{1}{2} \mathbf{B}_{\perp},
\end{equation}
where $\mathbf{U} \mathbf{U}^{\mathrm{T}} = \mathbf{I} - \hat{\mathbf{d}} \hat{\mathbf{d}}^{\mathrm{T}}$ is a projector onto the non-null space of $\mathcal{G}$. 
Projection with $\hat{\mathbf{d}} \hat{\mathbf{d}}^{\mathrm{T}}$ yields the remaining part of $\mathbf{Q}^{-1} \boldsymbol{\Lambda}$,
\begin{equation}
  \hat{\mathbf{d}} \hat{\mathbf{d}}^{\mathrm{T}} \mathbf{Q}^{-1} \boldsymbol{\Lambda} \approx \hat{\mathbf{d}} \hat{\mathbf{d}}^{\mathrm{T}} \mathbf{Q}^{-1} \sum_{\alpha} \boldsymbol{\Lambda}_{\alpha} - \frac{1}{2} \hat{\mathbf{d}} \hat{\mathbf{d}}^{\mathrm{T}} \mathbf{S}^{-1} \mathbf{U} (\mathbf{U}^{\mathrm{T}} \mathbf{S}^{-1} \mathbf{U})^{-1}  \mathbf{U}^{\mathrm{T}} \mathbf{B},
\end{equation}
which is also an intensive quantity. We also note a type of additivity for the projected \magm{intrinsic magnetic moment},
\begin{equation}
  \hat{\mathbf{d}} \hat{\mathbf{d}}^{\mathrm{T}} \boldsymbol{\Lambda} = \hat{\mathbf{d}} \hat{\mathbf{d}}^{\mathrm{T}} \sum_{\alpha} \boldsymbol{\Lambda}_{\alpha}.
\end{equation}

Let us now consider {\it local} approximations to the exchange--correlation energy. An obvious first ansatz is the LDA-like form
\begin{equation}
  F^{\text{LDA}}_{\text{xc}}(\rho,\boldsymbol{\Lambda}) = \int f(\rho(\mathbf{r}),|\boldsymbol{\Lambda}|) d\mathbf{r},
\end{equation}
where only the magnitude $|\boldsymbol{\Lambda}|$ can enter without breaking rotational invariance. We assume that $f(\rho(\mathbf{r}),|\boldsymbol{\Lambda}|)$ is negligible whenever $\rho(\mathbf{r})$ is negligible. Exploiting the fact that the subsystems are separated and requiring additive separability, we now obtain
\begin{equation}
  F^{\text{LDA}}_{\text{xc}}(\rho,\boldsymbol{\Lambda}) = \sum_{\alpha} \int f(\rho_{\alpha}(\mathbf{r}),|\boldsymbol{\Lambda}|) d\mathbf{r} = \sum_{\alpha} \int f(\rho_{\alpha}(\mathbf{r}),|\boldsymbol{\Lambda}_{\alpha}|) d\mathbf{r}.
\end{equation}
This needs to hold for a large class of ground-state densities $\rho_{\alpha}$. In practice, we expect equality to hold also for the integrands, leading to
\begin{equation}
  f(\rho_{\alpha}(\mathbf{r}),|\boldsymbol{\Lambda}|) = f(\rho_{\alpha}(\mathbf{r}),|\boldsymbol{\Lambda}_{\alpha}|).
\end{equation}
We have seen above that the total $|\boldsymbol{\Lambda}|$ diverges to infinity in the limit of infinite subsystem separation. Consequently, the above relation can only be satisfied when $f$ is independent of the second argument.

A slightly more refined ansatz relies on the quantity $\boldsymbol{\beta} = -2\mathbf{Q}^{-1} \boldsymbol{\Lambda}$ instead of the \magm{intrinsic magnetic moment} directly. Note that $\mathbf{Q}^{-1}$ is a non-local but tractable functional of the density. An LDA-like form is then given by
\begin{equation}
  F^{\text{LDA}}_{\text{xc}}(\rho,\boldsymbol{\beta}) = \int f(\rho(\mathbf{r}),|\boldsymbol{\beta}|) d\mathbf{r}.
\end{equation}
Additive separability now leads to the condition
\begin{equation}
  f(\rho_{\alpha}(\mathbf{r}),|\boldsymbol{\beta}|) = f(\rho_{\alpha}(\mathbf{r}),|\boldsymbol{\beta}_{\alpha}|),
\end{equation}
where $\boldsymbol{\beta} = \mathbf{B}$ or
$\boldsymbol{\beta} = \mathbf{B}_{\perp}$ depending on whether $\mathcal{G}$
has full rank or not. Again, additive separability cannot be satisfied unless $f$
is independent of its second argument. However, the prospects are
better for approximate additive separability since
$\boldsymbol{\beta} \approx \boldsymbol{\beta}_{\alpha}$ holds whenever the size or properties of the subsystems prevent large effective angular velocities $\boldsymbol{\omega}_{\alpha}$.

A slightly more flexible LDA-like form is obtained by giving the functional access to data about the global dimensions of the system via the moment-of-inertia tensor,
\begin{equation}
  F^{\text{LDA}}_{\text{xc}}(\rho,\mathbf{Q},\boldsymbol{\Lambda}) = \int f(\rho(\mathbf{r}),\mathbf{Q},\boldsymbol{\Lambda}) d\mathbf{r}.
\end{equation}
With this form, it becomes possible, for example, to project $\boldsymbol{\Lambda}$ onto the principal components of $\mathbf{Q}$ and to treat the components differently 
depending on the corresponding eigenvalue of $\mathbf{Q}$.

\section{The physical current density and the physical magnetic moment}
\label{secPSCRITIQUE}

After the seminal work by Vignale and Rasolt, paramagnetic CDFT has become the dominant generalization of standard DFT to allow for arbitrary, external magnetic fields. 
It is a {\it prima facie} surprising feature that the paramagnetic current density appears as a basic variable in this formulation.  A few works have attempted to construct an alternative formulation based on the physical current density~\cite{DIENER_JP3_9417,PAN_IJQC110_2833,PAN_JPCS73_630,SAHNI_PRA85_052502,PAN_PRA86_042502}. However, none of the three foundations of DFT---that is, the HK theorem, the constrained-search approach, or Lieb's convex analysis framework---has been rigorously extended to physical current densities. 

The existence of an adequate HK-like mapping from the physical density pair $(\rho,\mathbf{j})$ to external potentials $(v,\mathbf{A})$ (or an equivalence class of potentials that differ by a gauge transformation) is presently an open question~\cite{TELLGREN_PRA86_062506,LAESTADIUS_IJQC114_782}. The existence of a HK functional of the physical current density thus cannot be excluded. However, such a functional would require a different approach than in standard DFT since it would not have a straightforward variation principle~\cite{VIGNALE_IJQC113_1422,TELLGREN_PRA86_062506,LAESTADIUS_PRA91_032508}.

A constrained-search approach suffers from the fact the inner minimization would be performed over a domain that depends on the external vector potential and the constrained-search functional would therefore not be universal, losing most of its formal and practical appeal. Consequently, we should also be careful in generalizing the common constrained-search notation `$\psi\mapsto \rho$' or `$\psi\mapsto \rho,\jmvec$' to a version with physical currents, `$\psi\mapsto \rho,\mathbf{j}$', as this leaves implicit the dependence of the wave-function domain on the vector potential. 

A generalization of the Lieb formalism with the physical density as a variable 
would require the identification of a Legendre--Fenchel transformation pair of a concave ground-state energy functional $E(v,\mathbf{A})$ and a convex universal functional $F(\rho,\mathbf{j})$. A necessary condition for this is that external potentials are paired linearly with densities. However, $E(v,\mathbf{A})$ is not concave in $\mathbf{A}$ and the interaction with the magnetic vector potential contains both linear and quadratic terms. An attempt to solve these problems using a change of variables $w := v - \frac{1}{2} A^2$ (in analogy with how $u := v + \frac{1}{2} A^2$ can be introduced in paramagnetic CDFT) leads to a direct contradiction (see Sec.~IV.D in Ref.~\cite{TELLGREN_PRA86_062506}). Counterexamples show that a change of variables $u_{\lambda} := v + \lambda A^2$ with $\lambda < \frac{1}{2} A^2$ cannot lead to a concave energy functional $E(u_{\lambda},\mathbf{A})$. In summary, a satisfactory CDFT featuring the physical current density as a variable is neither presently available nor completely excluded.

Alongside the LDFT formalism presented above, we can imagine a formalism based instead on the \magm{physical magnetic moment}. 
However, the difficulties for a physical CDFT would also be present for a physical LDFT. For example, the constrained-search domain indicated by `$\psi \mapsto \rho, \mathbf{J}$' would depend on the external potential, and therefore be non-universal. LDFT thus provides a simplified setting where such difficulties can be analysed and possibly circumvented. To date, no physical LDFT has been proposed or analysed in the literature. A hybrid formalism, based on the triple $(\rho,\mathbf{L},\mathbf{j})$, has been proposed~\cite{PAN_JCP143_174105}. We now compare this approach with LDFT formalism.

\subsection{Comparison of LDFT with a recent approach using $(\rho,\mathbf{L},\mathbf{j})$}

In Ref.~\onlinecite{PAN_JCP143_174105} a HK-like result is introduced in a setting where only uniform magnetic fields and gauge-fixed vector potentials are allowed. After minor modifications, the result can be stated as follows:
\begin{theorem}
  Let two or more potential pairs $(v_k,\mathbf{B}_k)$, with $k = 1, \ldots, K \geq 2$, be given. Let the vector potentials be gauge fixed to the form $\mathbf{A}_k = \frac{1}{2} \mathbf{B}_k \times \mathbf{r}_{\mathbf{G}}$.
 Suppose that $\psi_k$ is the non-degenerate ground state of $H(v_k,\mathbf{A}_k)$ and that all $\psi_k$ give rise to the same density triple $(\rho_{\psi_k}, \mathbf{L}_{\mathbf{G};\psi_k}, \mathbf{j}_k) = (\rho, \mathbf{L}_{\mathbf{G}}, \mathbf{j})$.
 Then, for all $k,l$, it follows that (a) $\mathbf{B}_k = \mathbf{B}_l$, (b) the scalar potentials differ by at most a constant $v_k = v_l + \textup{const}$, and (c) $\psi_k = \psi_l$.
\end{theorem}
Because the physical current density operator is system dependent, it is important to note that
\begin{equation}
  \mathbf{j}_k = \jmvecix{\psi_k} + \rho \mathbf{A}_k,
\end{equation}
and that $\jmvecix{\psi_k}$ is not independent of $k$.

The theorem above is the point of departure for the density-functional formalism proposed in Ref.~\onlinecite{PAN_JCP143_174105}. In this setting, the magnetic vector potentials are restricted to a three-dimensional family, whereas the physical current densities belong to an infinite-dimensional function space, hinting
at considerable redundancy of the formalism. Indeed, the allowed family of vector potentials is the subspace of $\ALinSpace$, defined in Eq.~\eqref{eqALinSpace}, corresponding to $\mathbf{a}=\mathbf{0}$. Consequently, the results in Sec.~\ref{secLDFT_HK} can be viewed as sharpened versions of the above theorem. Firstly, the current density determines the \magm{physical magnetic moment} $\mathbf{J}_{\mathbf{R}}$. In turn, the triple $(\rho,\boldsymbol{\Lambda},\mathbf{J}_{\mathbf{R}})$ determines the pair $(v,\mathbf{B})$ {\it even without the assumption that $\psi_k$ is a ground state}. An explicit expression for $\mathbf{B}$ is obtained from Eq.~\eqref{eqBfromLJ} above. Secondly, when the ground-state assumption is made, the physical current density and \magm{physical magnetic moment} can be dropped altogether for general non-cylindrically symmetric densities since $(\rho,\mathbf{L}_{\mathbf{G}})$ is sufficient to determine $(v,\mathbf{B})$. For cylindrically symmetric ground-state densities, some information additional to $(\rho,\mathbf{L}_{\mathbf{G}})$ is needed to determine the component of $\mathbf{B}$ along the symmetry axis. Even in this case, however, the density pair $(\rho,\mathbf{L}_{\mathbf{G}})$ is sufficient to determine the ground-state wave function and therefore the exchange--correlation energy and expectation values of operators independent of the vector potential (i.e., $\boldsymbol{\pi} = \mathbf{p}+\mathbf{A}$ and $\mathbf{r}\times\boldsymbol{\pi}$ are not included).

In Ref.~\onlinecite{PAN_JCP143_174105}, a constrained-search functional that incorporates a new constraint that the wave function must reproduce a given canonical momentum is given, following a similar line of reasoning to earlier works by the same authors. It is unclear whether or not their constrained-search functional is intended to be a {\it universal} functional~\cite{VIGNALE_IJQC113_1422,TELLGREN_PRA86_062506,LAESTADIUS_PRA91_032508}. A density functional is said to be universal if it is independent of the external potentials $(v,\mathbf{A})$. This definition is the natural generalization of the original usage by Hohenberg and Kohn~\cite{HOHENBERG_PR136_864} and common usage in textbooks and the literature (e.g., in Refs.~\onlinecite{PARR89,DREIZLER90}, in a DFT setting where $\mathbf{A}=\mathbf{0}$ is implicitly assumed). The term has also been used in the same manner by authors in the context of CDFT or BDFT (e.g., in Refs.~\onlinecite{ERHARD_PRA53_R5,GRAYCE_PRA50_3089}). 

Addressing this ambiguity in the use of the term `universal', the authors of Ref.~\onlinecite{PAN_JCP143_174105} have emphasized the fact that the external potentials are fixed in the Rayleigh--Ritz variation principle~\cite{PAN_IJQC113_1424,PAN_IJQC114_233,PAN_JCP143_174105}. This answers criticism of their previous work based on inconsistency between the universality of their functional and the ordinary variation principle. Specifically, the derivations are only correct for a semi-universal functional $F(\rho,\mathbf{L}_{\mathbf{G}},\mathbf{j},\mathbf{A})$ that is universal with respect to $v$ and non-universal with respect to $\mathbf{A}$ or $\mathbf{B}$. 

Nonetheless, the term universal is still used to refer to their constrained-search functional (above their Eq.~(45) in Ref.~\onlinecite{PAN_JCP143_174105}). It is possible that `universal' in their terminology corresponds to `semi-universal' in our terminology. However, insofar as the intent is to derive a density-functional that assigns the same intrinsic energy to a triple $(\rho,\mathbf{L},\mathbf{j})$ irrespective of the external potential pair $(v,\mathbf{A})$ their derivation is mistaken. Criticism directed at their previous work still applies, largely unaffected by the new constraint on the \magm{paramagnetic moment}: The search domain of the functional defined in their Eq.~(45), indicated by `$\Psi_{\rho,\mathbf{j}}(N,\mathbf{L}) \mapsto \rho,\mathbf{j}$', must depend on the external vector potential and therefore be non-universal. Moreover, universality of a functional of $(\rho,\mathbf{j})$ has been shown to be inconsistent with the usual Rayleigh--Ritz variation principle~\cite{VIGNALE_IJQC113_1422,TELLGREN_PRA86_062506,LAESTADIUS_PRA91_032508}, and this holds also for a universal functional of $(\rho,\mathbf{L},\mathbf{j})$. 
\subsubsection{A non-universal constrained-search functional of $(\rho,\mathbf{j},\mathbf{A})$}

Let us be careful in tracking not only the overt dependence of the
Hamiltonian on the external potential, but also the dependence of the
various search domains on the potentials. We retain the choice of
function spaces described above since they guarantee that the physical
kinetic energy can split into finite canonical, paramagnetic, and
diamagnetic terms and also that the \magm{magnetic moment} is always finite.
Because the formalism of Ref.~\cite{PAN_JCP143_174105} features a family of
functionals parametrized by prescribed values of
$\mathbf{L}$, these conditions cannot be avoided.

We now note that the minimum energy for a given value of $\mathbf{L}$ can be written as
\begin{equation}
 \begin{split}
	E_{\mathbf{L}}(v,\mathbf{A}) = \inf_{\substack{\psi \in \WMSobolev{\mathbf{0}} \\ \psi \mapsto \mathbf{L}}} \bra{\psi} H(v,\mathbf{A}) \ket{\psi} &= \inf_{\substack{(\rho,\jmvec)\in X\times Y \\ \int \mathbf{r}\times \jmvec d\mathbf{r} = \mathbf{L}}}
	\Big( (\rho|v + \tfrac{1}{2} A^2) + (\jmvec|\mathbf{A}) +  \inf_{\substack{\psi \in \WMSobolev{\mathbf{0}} \\ \psi \mapsto \rho, \jmvec}} \bra{\psi} \op{H}_0 \ket{\psi}\Big)\\
	& = \inf_{\substack{(\rho,\jmvec)\in X\times Y \\ \int \mathbf{r}\times \jmvec d\mathbf{r} = \mathbf{L}}}
	\Big( (\rho|v + \tfrac{1}{2} A^2) + (\jmvec|\mathbf{A}) +  \Fvr(\rho,\jmvec)\Big),
 \end{split}
\end{equation} 
where Vignale and Rasolt's constrained-search functional $\Fvr(\rho,\jmvec)$ is defined as in Sec.~\ref{secCDFTREVIEW} above.
We reparametrize the outer minimization by setting $\mathbf{j} = \jmvec + \rho  \mathbf{A}$,
\begin{equation}
 \begin{split}
	E_{\mathbf{L}}(v,\mathbf{A}) = \inf_{\substack{(\rho,\mathbf{j})\in R_{\mathbf{A}} \\ \int \mathbf{r}\times (\mathbf{j} - \rho \mathbf{A}) d\mathbf{r} = \mathbf{L}}}
	\Big( (\rho|v - \tfrac{1}{2} A^2) + (\mathbf{j}|\mathbf{A}) +  \Fvr(\rho,\mathbf{j} - \rho \mathbf{A})\Big),
 \end{split}
\end{equation} 
where $R_{\mathbf{A}} = \{(\rho,\mathbf{j}) | \mathbf{j} = \jmvec + \rho \mathbf{A}, \, (\rho,\jmvec)\in X \times Y \}$. Due to the choice of $Y$ as a weighted $L^1$ space, we see that $\rho \mathbf{A} \in Y$. Consequently, the set $R_{\mathbf{A}}$ is independent of $\mathbf{A}$ and, in fact, we have $R_{\mathbf{A}} = X \times Y$. We can now write
\begin{equation}
 \begin{split}
	E_{\mathbf{L}}(v,\mathbf{A}) = \inf_{\substack{(\rho,\mathbf{j})\in X\times Y \\ \int \mathbf{r}\times (\mathbf{j} - \rho \mathbf{A}) d\mathbf{r} = \mathbf{L}}}
	\Big( (\rho|v - \tfrac{1}{2} A^2) + (\mathbf{j}|\mathbf{A}) +  \Fvr(\rho,\mathbf{j} - \rho \mathbf{A})\Big).
 \end{split}
\end{equation} 
Thus, by defining the non-universal functional $F_{\text{PS}}(\rho,\mathbf{j},\mathbf{A}) := \Fvr(\rho, \mathbf{j} -\rho \mathbf{A})$ we have
\begin{equation}
 \begin{split}
	E_{\mathbf{L}}(v,\mathbf{A}) = \inf_{\substack{(\rho,\mathbf{j})\in X\times Y \\ \int \mathbf{r}\times (\mathbf{j} - \rho \mathbf{A}) d\mathbf{r} = \mathbf{L}}}
	\Big( (\rho|v - \tfrac{1}{2} A^2) + (\mathbf{j}|\mathbf{A}) +  F_{\text{PS}}(\rho,\mathbf{j}, \mathbf{A}) \Big).
 \end{split}
\end{equation} 
Since this is simply a reformulation of Vignale and Rasolt's CDFT, it is also correct. We note, however, that the original formulation due Vignale and Rasolt yields a universal functional of $(\rho,\jmvec)$---that is, the external potentials appear {\it only} in the integral pairings $(\jmvec|\mathbf{A})$ and $(\rho|A^2)$. By contrast, after the reformulation, the external potential also appears in the search domain for $F_{\text{PS}}(\rho,\mathbf{j},\mathbf{A})$ and, via the constraint $\mathbf{L} = \int \mathbf{r}\times(\mathbf{j}-\rho\mathbf{A}) d\mathbf{r}$, in the domain for the outer minimization. Different vector potentials $\mathbf{A}$ give rise to different search domains $\{ \psi | \psi \in \MSobolev{\mathbf{0}}, \rho_{\psi} = \rho, \jmvecix{\psi} = \mathbf{j} - \rho \mathbf{A} \}$ for the functional $F_{\text{PS}}(\rho,\mathbf{j},\mathbf{A})$.

The functional $F_{\text{PS}}(\rho,\mathbf{j},\mathbf{A})$ can  interpreted as a non-universal functional with potentials that remain fixed in the variation principle, rationalizing the dependence of the search domains on the external potential $\mathbf{A}$. However, the advantage of such a formulation is unclear. Consider for example a light modification of the BDFT formalism due to Grayce and Harris. If universality is not required, we may also write
\begin{equation}
 \begin{split}
	E_{\mathbf{L}}(v;\mathbf{A}) & = \inf_{\substack{\psi \in \MSobolev{\mathbf{0}} \\ \psi \mapsto \mathbf{L}}} \bra{\psi} H(v,\mathbf{A}) \ket{\psi}  = \inf_{\rho \in X}	\Big( (\rho|v) + \inf_{\substack{\psi \in \MSobolev{\mathbf{0}} \\ \psi \mapsto \mathbf{L}}} \bra{\psi} H(0,\mathbf{A}) \ket{\psi}\Big)  \\
	& = \inf_{\rho \in X} \Big( (\rho|v) + F_{\text{GH};\mathbf{L}}(\rho;\mathbf{A})\Big).
 \end{split}
\end{equation} 
Grayce and Harris' original work constructed a family of functionals parametrized by an arbitrary magnetic field $\mathbf{B}(\mathbf{r})$~\cite{GRAYCE_PRA50_3089}. Above we have modified this to a functional $F_{\text{GH};\mathbf{L}}(\rho;\mathbf{A})$ parametrized by $\mathbf{A}(\mathbf{r})$ which also incorporates the restriction on $\mathbf{L}$. From the above nested minimization, it is clear that the pair $(\rho,\mathbf{A})$, or alternatively the pair $(\rho,\mathbf{B})$, alone determines the energy additional to the electrostatic interaction $(\rho|v)$ with the external electric potential. In particular, the exchange--correlation energy is determined by $(\rho,\mathbf{A})$. By contrast, the functional $F_{\text{PS}}(\rho,\mathbf{j},\mathbf{A})$ depends on the additional vector field $\mathbf{j}$, yet $\mathbf{j}$ is formally redundant when $(\rho,\mathbf{A})$ are given.

\subsubsection{Non-standard properties of a universal functional of $(\rho,\mathbf{j},\mathbf{A})$}

Suppose $H(v_k,\mathbf{A}_k)$ has the ground state $\psi_k$, with
density triple $(\rho_k,\mathbf{j}_k,\mathbf{L}_k)$ where
$\mathbf{j}_k = \jmvecix{k} + \rho_k \mathbf{A}_k$. We let
$\mathcal{A}_N$ denote the set of all such $(v,\mathbf{A})$-representable ground-state density triples and $\mathcal{A}_{N,\mathbf{L}}$ is the subset with a given \magm{paramagnetic moment}. To formulate a variation principle for such a density triple, we begin by rewriting the energy assigned to a state $\psi_l$ by the Hamiltonian $H(v_k,\mathbf{A}_k)$,
\begin{equation}
  \bra{\psi_l} H(v_k,\mathbf{A}_k) \ket{\psi_l}  = \bra{\psi_l} \op{H}_0 \ket{\psi_l} + \frac{1}{2} \mathbf{B}_k\cdot\mathbf{L}_l + (\rho_l|v_k+\tfrac{1}{2} A_k^2).
\end{equation}
Alternatively, taking care to note that $\mathbf{j}_l = \jmvecix{l} + \rho_l \mathbf{A}_l \neq \jmvecix{l} + \rho_l \mathbf{A}_k$, unless $\mathbf{A}_k = \mathbf{A}_l$, we can also write
\begin{equation}
 \begin{split}
  \bra{\psi_l} \op{H}(v_k,\mathbf{A}_k) \ket{\psi_l}
   & = \bra{\psi_l} \op{H}_0 \ket{\psi_l} + (\mathbf{j}_l|\mathbf{A}_k) + (\rho_l|v_k+\tfrac{1}{2} A_k^2-\mathbf{A}_k\cdot\mathbf{A}_l).
 \end{split}
\end{equation}
Suppose now that there exists a universal functional $F_{\mathrm{PS}}$ with the property
\begin{equation}
  F_{\mathrm{PS}}(\rho_l,\mathbf{j}_l,\mathbf{L}_l) = \bra{\psi_l} \op{H}_0 \ket{\psi_l}
\end{equation}
for arbitrary $\psi_l \mapsto (\rho_l,\mathbf{j}_l,\mathbf{L}_l)$. 
The LDFT functional $\Fldft(\rho,\mathbf{L})$, specialized to the present gauge restrictions (allowing the variable $\mathbf{p}$ to be dropped) is an example of such a functional.
In what follows, we assume that there are also universal functionals with a non-trivial $\mathbf{j}$-dependence. Identifying $F_{\mathrm{PS}}(\rho_l,\mathbf{j}_l,\mathbf{L}_l)$ and taking affine combinations of the two alternative expressions now yields,
\begin{equation}
  \label{eqGVARdef}
 \begin{split}
  \bra{\psi_l} H(v_k,\mathbf{A}_k) \ket{\psi_l} & = F_{\mathrm{PS}}(\rho_l,\mathbf{j}_l,\mathbf{L}_l) + \frac{1}{2} \lambda \mathbf{B}_k\cdot\mathbf{L}_l + \bar{\lambda} (\mathbf{j}_l|\mathbf{A}_k) + (\rho_l|v_k+\tfrac{1}{2} A_k^2- \bar{\lambda} \mathbf{A}_k\cdot\mathbf{A}_l) =: \mathcal{E}_{\lambda}(v_k,\mathbf{A}_k,\rho_l,\mathbf{j}_l,\mathbf{L}_l),
 \end{split}
\end{equation}
where $\lambda \in \mathbb{R}$ is an arbitrary parameter and $\bar{\lambda} = 1 - \lambda$. The middle expression defines a functional $\mathcal{E}_{\lambda}$ of $(v_k,\mathbf{A}_k,\rho_l,\mathbf{j}_l,\mathbf{L}_l)$. We also have
\begin{align}
   \mathcal{E}_{\lambda}(v_k,\mathbf{A}_k,\rho_l,\mathbf{j}_l,\mathbf{L}_l) & \geq E(v_k,\mathbf{A}_k),
\end{align}
with equality when $k=l$. This inequality is just an alternative form
of $\bra{\psi_l} H(v_k,\mathbf{A}_k) \ket{\psi_l} \geq \bra{\psi_k}
H(v_k,\mathbf{A}_k) \ket{\psi_k}$, which is just the Rayleigh--Ritz
variation principle for wave functions. Identifying $\bra{\psi_k} H(v_k,\mathbf{A}_k) \ket{\psi_k}$ as $E(v_k,\mathbf{A}_k)$, we now obtain the variation principle
\begin{align}
   \label{eqLBVarPrinc}
   \min_{(\rho_l,\mathbf{j}_l,\mathbf{L}_l) \in \mathcal{A}_N} \mathcal{E}_{\lambda}(v_k,\mathbf{A}_k,\rho_l,\mathbf{j}_l,\mathbf{L}_l) & = E(v_k,\mathbf{A}_k).
\end{align}
(We here retain the index $l$ to distinguish density triples based on the potentials they arise from. However, the minimization is of course performed with $(\rho_l,\mathbf{j}_l,\mathbf{L}_l)$ ranging over all $(v,\mathbf{A})$-representable triples and not just a countable subset.)
By constraining all \magm{magnetic moments} to a particular value $\mathbf{L}_k = \mathbf{L}$, we also obtain a variation principle for the lowest energy compatible with this constraint
\begin{align}
   \label{eqLBVarPrincL}
   \min_{(\rho_l,\mathbf{j}_l) \in \mathcal{A}_{N,\mathbf{L}}} \mathcal{E}_{\lambda}(v_k,\mathbf{A}_k,\rho_l,\mathbf{j}_l,\mathbf{L}) & = E_{\mathbf{L}}(v_k,\mathbf{A}_k).
\end{align}

The case $\lambda = 1$ corresponds closely to the LDFT setting, at
least if $F$ is independent of $\mathbf{j}_l$. In this case, the
energy functional being minimized is independent of
$\mathbf{j}_l$. Hence, the minimization over $\mathbf{j}_l$ is
redundant, but does not lead to serious complications of the
theory. The case $\lambda = 0$ corresponds to a setting akin to a physical
CDFT. In this case, the appearance of the term
$-\bar{\lambda} (\rho_l|\mathbf{A}_k\cdot\mathbf{A}_l)$ in
Eq.~\eqref{eqGVARdef} gives the variation principle a non-standard
form. The minimization over $\mathbf{j}_l$ cannot be performed without
access to the vector potential $\mathbf{A}_l$, which in general is
different from the external potential $\mathbf{A}_k$ that appears in
the particular Hamiltonian $H(v_k,\mathbf{A}_k)$ used to assign
energies. In this regard, the constraint on the \magm{paramagnetic
moment} and the restriction of vector potentials to a
three-dimensional space makes a crucial difference. Without the
restrictions, a mapping $(\rho_l,\mathbf{j}_l) \mapsto \mathbf{A}_l$
would be required, but there is no rigorous existence proof for this
mapping and the mapping is unlikely to have a simple, practical
form. With the restrictions, however, there is a simple, explicit
mapping $(\rho_l,\mathbf{j}_l,\mathbf{L}_l) \mapsto \mathbf{A}_l$ given
by Eq.~\eqref{eqBfromLJ} above. A more subtle point is that the
mapping $(\rho_l,\mathbf{j}_l,\mathbf{L}_l) \mapsto \mathbf{A}_l$ makes
no use of the ground-state assumption, and is not restricted to
$(v,\mathbf{A})$-representable triples $(\rho,\mathbf{j},\mathbf{L})$.

From the above discussion, we see that the non-standard variation
principles Eqs.~\eqref{eqLBVarPrinc}--\eqref{eqLBVarPrincL} delivers the correct
ground-state energy and minimum energy compatible with $\mathbf{L}$, respectively. However, the minimizing physical current density
is not necessarily the same as that arising from the ground
state. We now adapt a result from Laestadius and
Benedicks~\cite{LAESTADIUS_PRA91_032508} to the present setting:
\begin{theorem} \label{Thm:nonUmin}
	Let $H$ model a one-particle system. There exist potentials $(v_k,\mathbf{A}_k)$ such that: (a) the minimizer $(\rho_l,\mathbf{j}_l,\mathbf{L}_l)$ of Eq.~\eqref{eqLBVarPrinc} is not unique, and (b) the minimizer $(\rho_l,\mathbf{j}_l)$ of Eq.~\eqref{eqLBVarPrincL} is not unique. Moreover, infinitely many of the minimizing $\mathbf{j}_l$ are not the physical current of a ground state of $H(v_k,\mathbf{A}_k)$---that is, $\mathbf{j}_l \neq \mathbf{j}_k = \jmvecix{k} + \rho_k \mathbf{A}_k$.
\end{theorem}
The proof of the above theorem follows every step in Sec.III.A of Ref.~\onlinecite{LAESTADIUS_PRA91_032508}, which still applies and where we let the minimizing $\mathbf{L}_l = \langle \psi_k| \op{\mathbf{L}} | \psi_k\rangle$.
Note that $\psi_k$ and the infinitely many current densities $\{\mathbf{j}_l\}$ are denoted $\psi_0$ and $\{\mathbf{j}_\varepsilon\}_{\varepsilon >0}$, respectively, in Ref.~\onlinecite{LAESTADIUS_PRA91_032508}. 

Turning next to the interpretation of constrained-search functional in Ref.~\onlinecite{PAN_JCP143_174105} as a universal functional, we note that it relies on a variation principle related to a functional
\begin{equation}
	\mathcal{E}_{\mathrm{PS}}(v_k,\mathbf{A}_k,\rho_l,\mathbf{j}_l,\mathbf{L}) = F(\rho_l,\mathbf{j}_l,\mathbf{L}) + (\mathbf{j}_l|\mathbf{A}_k) + (\rho_l|v_k-\tfrac{1}{2} A_k^2).
\end{equation}
The relation to the functional $\mathcal{E}_1$ studied above is
\begin{equation}
	\mathcal{E}_{\mathrm{PS}}(v_k,\mathbf{A}_k,\rho_l,\mathbf{j}_l,\mathbf{L}_l) = \mathcal{E}_1(v_k,\mathbf{A}_k,\rho_l,\mathbf{j}_l,\mathbf{L}_l) - (\rho|A_k^2 - \mathbf{A}_k\cdot\mathbf{A}_l).
\end{equation}
Minimization of $\mathcal{E}_{\mathrm{PS}}$ with respect to $(\rho_l,\mathbf{j}_l)$ corresponds to a seemingly natural variation principle. Unfortunately, this variation principle does not always deliver the correct energy. Again we adapt a result from Laestadius and
Benedicks~\cite{LAESTADIUS_PRA91_032508} to the present setting:
\begin{theorem}
	There exist potentials $(v_k,\mathbf{A}_k)$ such that:
	\begin{equation} \label{Eq:17feb1}
		\min_{(\rho_l,\mathbf{j}_l,\mathbf{L}_l) \in \mathcal{A}_N} \mathcal{E}_{\mathrm{PS}}(v_k,\mathbf{A}_k,\rho_l,\mathbf{j}_l,\mathbf{L}_l) < E(v_k,\mathbf{A}_k).
	\end{equation}
	There also exists $\mathbf{L}$ such that:
	\begin{equation} \label{Eq:17feb2}
		\min_{(\rho_l,\mathbf{j}_l) \in \mathcal{A}_{N,\mathbf{L}}} \mathcal{E}_{\mathrm{PS}}(v_k,\mathbf{A}_k,\rho_l,\mathbf{j}_l,\mathbf{L}) < E_{\mathbf{L}}(v_k,\mathbf{A}_k).
	\end{equation}
\end{theorem}
The essential steps of the proof of this statement is given in Sec.II of the work by Laestadius and Benedicks~\cite{LAESTADIUS_PRA91_032508}. For the first part above, Eq.~\eqref{Eq:17feb1}, we just need to add the remark given after Theorem \ref{Thm:nonUmin} regarding $\mathbf{L}_l$. For the second part, Eq.~\eqref{Eq:17feb2}, we set $\mathbf{L} = \langle \psi_k| \hat{\mathbf{L}}|\psi_k\rangle$, where $\psi_k$ is the ground state of $H(v_k,\mathbf{A}_k)$.

\section{Discussion and conclusion}
\label{secCONCLUSION}

A specialization of paramagnetic CDFT to uniform magnetic fields, represented by linear magnetic vector potentials, has been presented. As a result of the restriction on the vector potentials, the basic variables are reduced to the electron density, the canonical momentum, and \magm{paramagnetic moment}. This theory, termed LDFT, admits a constrained-search formulation as well as a convex formulation analogous to Lieb's formulation of standard DFT. Unlike CDFT, LDFT also admits strong analogues of the HK theorem whenever the density is not rotationally symmetric: In this case, the ground-state triple $(\rho,\mathbf{p},\mathbf{L}_{\mathbf{G}})$ determines the scalar potential $u$ and the linear magnetic vector potential $\mathbf{A}$. Moreover, the ground-state pair $(\rho,\boldsymbol{\Lambda})$ of density and \magm{intrinsic magnetic moment} determines the scalar potential and magnetic field $(u,\mathbf{B})$. For a large total system, consisting of well-separated subsystems, the magnetic part of the HK-mapping can be written in a simple explicit form involving the inverse moment-of-inertia tensor.

Many issues in CDFT find their analogue, in a simplified setting, in LDFT. For example, both the CDFT and LDFT universal functionals can be explored using Lieb optimization. The optimization over the scalar potential is known to be computationally and numerically demanding. In practice, the space of potentials $X^*$ must be replaced by a subset spanned by a finite number of basis functions, typically Gaussian-type functions. The optimization over vector potentials $\mathbf{A} \in Y^*$ exacerbate  these problems. 
By contrast, in the LDFT setting, the optimization over linear vector potentials $\mathbf{A}\in \ALinSpace$ requires only six additional variational parameters and is 
straightforward compared with the scalar potential optimization. Another example is that gauge invariance of the exchange--correlation energy leads to a CDFT functional that depends on the vorticity, whereas \emph{gauge-shift} invariance leads to an LDFT functional that depends on the \magm{intrinsic magnetic moment}. In the special case that the current density corresponds to rigid rotation, the vorticity is constant throughout space and related to \magm{intrinsic magnetic moment} through $\boldsymbol{\Lambda} = \frac{1}{2} \mathbf{Q} \boldsymbol{\nu}$. In more realistic cases, the two quantities have rather different properties since $\boldsymbol{\nu}$ is a local, intensive quantity and $\boldsymbol{\Lambda}$ is a global quantity that grows superlinearly with system size. This difference manifests itself in, for example, the failure of simple LDA-like forms to satisfy additive separability. Even with a more flexible ansatz, it is likely to be challenging to construct additively separable approximations. Here Lieb optimization has a large potential to be useful in uncovering the $\boldsymbol{\Lambda}$ dependence of the exchange--correlation energy.

A third issue is that of the possibility of replacing the paramagnetic current density by the physical current in CDFT. The LDFT analogue is to replace the \magm{paramagnetic moment} by the \magm{physical magnetic moment}. In many respects this would be appealing since both vorticity and \magm{intrinsic magnetic moment} have practical drawbacks. The challenges that have been pointed out for CDFT---namely, the dilemma that constrained-search formulations suffer from either non-universality or tension with the 
variation principle---carry over directly a hypothetical LDFT with \magm{physical magnetic moment}. Under the restriction of uniform magnetic fields and completely gauge-fixed vector potentials, a hybrid approach was recently suggested by Pan and Sahni, where the triple $(\rho,\mathbf{j},\boldsymbol{\Lambda})$ of density, physical current density, and \magm{paramagnetic moment} serves as the basic variables. However, in light of the HK-like results proved in this work, the Pan and Sahni formulation is seen to be heavily redundant. The pair $(\rho,\boldsymbol{\Lambda})$ is sufficient to determine the exchange--correlation energy. When $\rho$ is not rotationally symmetric, the pair $(\rho,\boldsymbol{\Lambda})$ also determines the external scalar potential and magnetic field. If a HK-like result for rotationally symmetric systems is desired, the triple $(\rho,\mathbf{J},\boldsymbol{\Lambda})$ of density, \magm{physical magnetic moment}, and \magm{intrinsic magnetic momentum} is sufficient.

\section*{Acknowledgements}

This work was supported by the Norwegian Research Council through the CoE Centre for Theoretical and Computational Chemistry (CTCC) Grant Nos.\ 179568/V30 and 171185/V30 and through the European Research Council under the European Union Seventh Framework Program through the Advanced Grant ABACUS, ERC Grant Agreement No. 267683. EIT was supported by the Norwegian Research Council through Grant No.~240674. AMT is grateful for support from a Royal Society University Research Fellowship and the Engineering and Physical Sciences Research Council EPSRC, Grant No.\ EP/M029131/1. SK and AL were supported by ERC-STG-2014 under grant agreement No.~639508.


%

\end{document}